\documentclass[a4paper,fleqn]{cas-dc}

\usepackage[authoryear,longnamesfirst]{natbib}

\usepackage{cuted}

\usepackage{amsthm,mathtools,bm}
\usepackage{tabularx,longtable,array}
\graphicspath{{figures/}}
\usepackage{tikz}
\usetikzlibrary{arrows.meta,positioning,fit,shapes.geometric}
\usepackage{enumitem}

\def\tsc#1{\csdef{#1}{\textsc{\lowercase{#1}}\xspace}}
\tsc{WGM}
\tsc{QE}
\newtheorem{theorem}{Theorem}
\newtheorem{proposition}{Proposition}

\newtheorem{corollary}{Corollary}

\newtheorem{assumption}{Assumption}

\newtheorem{remark}{Remark}

\newcommand{\X}{\mathcal{X}}
\newcommand{\A}{\mathcal{A}}

\newcommand{\Cfam}{\mathcal{C}}
\newcommand{\Hfam}{\mathcal{H}}
\newcommand{\Kfam}{\mathcal{K}}
\newcommand{\Pfam}{\mathcal{P}}
\newcommand{\Gdug}{\mathfrak{G}}
\newcommand{\Pp}{\mathbb{P}}
\newcommand{\E}{\mathbb{E}}
\newcommand{\R}{\mathbb{R}}

\newcommand{\F}{\mathbb{F}}
\newcommand{\dd}{\mathrm{d}}
\newcommand{\T}{\mathsf{T}}
\newcommand{\Id}{\mathrm{Id}}
\newcommand{\TV}{\mathrm{TV}}

\newcommand{\Dgm}{\mathrm{Dgm}}

\newcommand{\grid}[1]{#1\mathbin{\times}#1}

\newcolumntype{Y}{>{\raggedright\arraybackslash}X}

\begin{document}
\let\WriteBookmarks\relax
\def\floatpagepagefraction{1}
\def\textpagefraction{.001}

\shorttitle{Dynamical uncertainty geometry}
\shortauthors{G.~Tommei}

\title[mode = title]{Dynamical uncertainty geometry for nonlinear transport}

\author[1]{Giacomo Tommei}[orcid=0000-0002-9593-613X]
\cormark[1]
\ead{giacomo.tommei@unipi.it}
\credit{Conceptualization, Formal analysis, Methodology, Writing -- original draft, Writing -- review and editing}

\affiliation[1]{organization={Department of Mathematics, University of Pisa},
            addressline={Largo B. Pontecorvo 5},
            city={Pisa},
            postcode={56127},
            country={Italy}}

\cortext[1]{Corresponding author.}

\ead[url]{people.unipi.it/tommei}

\begin{abstract}
Uncertainty in nonlinear dynamical systems is often organized by transport structures that are not apparent from posterior geometry alone. We introduce Dynamical Uncertainty Geometry (DUG), a framework that separates three ingredients that are frequently conflated: posterior uncertainty, future observations, and model-induced transport. DUG indexes posterior credible sets and transport-conditioned subsets by a common enclosed-probability coordinate, allowing physically meaningful classes to be tracked across credibility levels rather than examined at a single threshold.

The framework combines a mass-ranked credible filtration, a description of future experiments through their induced probability laws and Fisher geometry, and a labelled transport skeleton representing dynamically distinct outcomes. We establish stability results for the resulting mass-indexed persistence modules under simultaneous perturbations of posterior density and probability measure, including continuum, refinement, and finite atomic formulations. These results provide quantitative control of persistence computed from numerical approximations and weighted grids.

Two benchmark problems illustrate the framework. In a short-arc orbit-determination setting, DUG identifies dynamically distinct return classes within a connected uncertainty region and highlights a difference between information-based and local Fisher-based observation-design criteria. In the Earth-Moon planar circular restricted three-body problem, DUG reveals multiple first-hit transport outcomes coexisting within a connected credible region and provides diagnostics for assessing their numerical resolution. Together, these examples show how topological summaries of posterior geometry, when coupled to transport labels and future experiments, yield a richer description of uncertainty than either posterior probabilities or dynamical classifications alone.

\end{abstract}




\maketitle

\section{Introduction}
\label{sec:introduction}

Predicting the future evolution of a dynamical system under uncertainty is a central task across the physical sciences. Whether the objective is orbit determination, transport prediction, or hazard assessment, scientific conclusions are drawn from distributions of plausible states rather than from individual trajectories. Nonlinear dynamics often organize transport through invariant manifolds, separatrices, and other phase-space structures that partition nearby initial conditions into qualitatively different futures. As a result, a connected region of uncertainty need not correspond to a single dynamical outcome. A posterior whose credible sets are connected at every mass level need not have futures organized by a single dynamical route: a stable manifold, separatrix, target preimage or committor ridge can divide one connected uncertainty set into regions reaching different destinations, and a declared future experiment may omit the event map defining those routes. The distinction between the \emph{shape of uncertainty}, the \emph{observations included in an experiment} and the \emph{organization of transport} is the central point of this paper.

It matters operationally: for example, in short arc Orbit Determination (OD) two states close in posterior coordinates can return to the same observation gate after different numbers of revolutions; in the Earth--Moon Planar Circular Restricted Three-Body Problem (PCRTBP) a connected posterior on a Poincar\'e section can straddle first-hit classes at the $L_1$ and $L_2$ necks. A posterior sample estimates probabilities in both cases, but the scientific conclusion needs more: one must identify credible subsets, attach physical labels, and test whether the resulting classes survive credibility and numerical refinement.

\begin{center}
\sffamily\small
\resizebox{0.98\linewidth}{!}{%
\begin{tikzpicture}[font=\small,>=Latex]
  \begin{scope}[xshift=0cm]
    \node[font=\bfseries] at (0,2.35) {(a) Posterior only};
    \draw[thick] (0,0) ellipse (1.75 and 1.12);
    \draw (0,0) ellipse (1.26 and 0.78);
    \draw (0,0) ellipse (0.74 and 0.42);
    \node at (0,-1.48) {$\beta_0(C^m)=1$ for all tested $m$};
    \node at (0,0) {connected uncertainty};
  \end{scope}
  \begin{scope}[xshift=4.8cm]
    \node[font=\bfseries] at (0,2.35) {(b) Add transport skeleton};
    \draw[thick] (0,0) ellipse (1.75 and 1.12);
    \draw (0,0) ellipse (1.26 and 0.78);
    \draw (0,0) ellipse (0.74 and 0.42);
    \draw[very thick] (-0.65,-1.02) .. controls (-0.28,-0.38) and (-0.12,0.42) .. (0.30,1.05);
    \node at (-0.78,0.25) {$L_1$};
    \node at (0.82,-0.10) {$L_2$};
    \node at (0,-1.48) {two labelled classes};
  \end{scope}
  \begin{scope}[xshift=9.6cm]
    \node[font=\bfseries] at (0,2.35) {(c) Scientific output};
    \node[draw,rounded corners,align=center,minimum width=3.25cm,minimum height=0.8cm] (p) at (0,0.65) {labelled class\\probabilities};
    \node[draw,rounded corners,align=center,minimum width=3.25cm,minimum height=0.8cm] (s) at (0,-0.65) {label-aware\\observation design};
    \draw[->,thick] (-2.02,0.20) -- (p.west);
    \draw[->,thick] (-2.02,-0.20) -- (s.west);
  \end{scope}
\end{tikzpicture}%
}
\refstepcounter{figure}\label{fig:intuition}
\par\medskip\noindent\textbf{Figure \thefigure:} {\footnotesize A posterior credible filtration can remain connected while a model-induced descriptor partitions it into labelled first-hit or transport classes. Transport-conditioned persistence checks that the separation is not a single threshold artifact.}
\end{center}

The Dynamical Uncertainty Geometry (DUG) is a mathematical object designed to make these choices explicit. Its data are grouped by provenance rather than assumed to be independent layers:
\begin{enumerate}[label=(\roman*)]
\item a \emph{posterior block} containing admissibility, the reference measure, the posterior, its mass-indexed credible filtration and the persistence of that filtration;
\item an \emph{experiment block} containing the full family of future output laws and its local Fisher pullback tensor;
\item a \emph{transport block} containing labelled model-induced descriptors, their incidences with credible sets and the resulting transport conditioned persistence modules.
\end{enumerate}
The three entries form an informational interface, not a canonical decomposition of the physical system: they may share the same dynamics and observation history, and if the output is a complete state trajectory then deterministic event labels are measurable functions of the experiment rather than primitives. Their separation becomes substantive exactly when the output is restricted by an independently specified sensor, cadence, noise model and data product. In particular, persistence of the posterior credible filtration is never conflated with persistence of transport-conditioned subsets.

Density superlevel sets already underlie the highest density regions and density cluster trees \cite{Hyndman1996,ChaudhuriDasgupta2010}. Distance to a measure (DTM) and DTM filtrations instead regularize distance functions by a mass parameter and obtain Wasserstein stability for geometric inference from noisy point clouds \cite{ChazalCohenSteinerMerigot2011,Anai2019}. Statistical work gives confidence sets for persistence diagrams and consistent homology recovery from estimated density level sets \cite{Fasy2014,BobrowskiMukherjeeTaylor2017}. DUG addresses a different perturbation: it retains density superlevel sets but replaces their raw density threshold by posterior rank, then controls simultaneous log-density and probability measure errors. Its output is also intersected with physically labelled transport sets. Transfer operator and set oriented methods identify coherent sets and approximate global structures \cite{DellnitzJunge2002,FroylandPadberg2009,Froyland2013,DellnitzKlusZiessler2017}; transition-path theory uses committors and reactive currents \cite{EVandenEijnden2006,EVE2010}. These methods can supply the transport entry to DUG, whose added output is their incidence with posterior mass and a declared future experiment.

The coupling supports statements unavailable from either side alone. Persistence of the posterior credible filtration can say that every credible set is connected, and transport analysis can identify $L_1$ and $L_2$ first-hit classes, but only their incidence module states whether both labels remain represented as posterior mass is enlarged. Likewise, only transport labels, posterior weights and conditional future laws jointly define a label discrimination utility such as $I(\Lambda;Z_\tau)$.

OD is the motivating class: covariance propagation, Admissible Regions (ARs), Lines and Manifolds of Variations (LOV, MOV), and Monte Carlo posteriors each solve important operational problems \cite{Milani2005,Tommei2007,MilaniGronchi2010,DelVigna2020,DelVigna2021}, but each retains a different encoding of the uncertainty: a covariance matrix does not determine global topology, a sampling manifold does not by itself define a calibrated ambient measure, and a particle cloud does not expose finite-time distinguishability unless the organizing structures are reconstructed.

\paragraph{What is new, and at which level.}
Since DUG assembles objects that exist separately, it is worth separating the claims by type.

\emph{Theorem level.} Theorem~\ref{thm:persistence-stability} is the one new mathematical result. It controls a filtration indexed by enclosed posterior mass under \emph{simultaneous} perturbation of the log-density index and of the measure that defines the index, and returns an interleaving in the mass coordinate itself. Standard stability controls a filtration indexed by a fixed function under perturbation of that function only \cite{CohenSteiner2007,Chazal2016}; DTM filtrations use a mass parameter to regularize a distance function and are stable in Wasserstein distance \cite{ChazalCohenSteinerMerigot2011,Anai2019}; cluster-tree and level set statistics control estimated level sets at fixed density thresholds \cite{ChaudhuriDasgupta2010,Fasy2014,BobrowskiMukherjeeTaylor2017}. None controls the reparameterization of the index by the estimated measure, which is exactly what a grid refinement changes. Corollary~\ref{cor:mass-rank-consistency} makes the bound vanish under refinement; Corollary~\ref{cor:atomic-mass-rank} removes atomlessness, so the statement applies verbatim to weighted grids.

\emph{Representation level.} The interface of Eq.~\eqref{eq:dug-operational} is a declaration discipline, not a theorem: posterior mass geometry, the future-law family and labelled transport are recorded separately with their provenance, so a reduction is checked against the task rather than assumed. Propositions~\ref{prop:mi-minimality} and \ref{thm:skeleton-obstruction} and Theorem~\ref{thm:covariance-obstruction} make the discipline non-vacuous by exhibiting models that agree on a reduced report and disagree on the scientific output.

\emph{Application level.} Three computed findings are new and are not consequences of the theorem: the transport-conditioned barcode of Section~\ref{sec:barcode} has exactly two essential bars with no component filtering; the observation-design calculation of Section~\ref{sec:design} yields two criteria whose epoch rankings are uncorrelated over a dense grid; and the ladder of Section~\ref{sec:pcr3bp-scaling} shows the unresolved PCRTBP boundary mass falling only like $h^{0.64}$, quantifying the interlaced first-hit boundary rather than merely bounding it.

Section~\ref{sec:setting} states the setting and notation. Section~\ref{sec:definition} defines the mass filtration, gives the continuum, refinement and atomic stability results, and states the operational interface. Sections~\ref{sec:topology-change}--\ref{sec:representation-diagnostics} contain the supporting analysis and the diagnostics, including an explicit statement in Section~\ref{sec:theorem-scope} of which computed quantity each result does and does not certify. Sections~\ref{sec:classical}--\ref{sec:computation} position the construction and define the discrete report. Sections~\ref{sec:numerical} and \ref{sec:pcr3bp} give the two benchmarks. Supporting covariance, coordinate and factorization arguments are consolidated in the appendices.

\section{Mathematical setting}
\label{sec:setting}

\subsection{State dynamics, admissibility and observations}

Let \(\X\) be a connected, second-countable smooth $n$-manifold, possibly with boundary, with Borel $\sigma$-algebra \(\mathcal B(\X)\). At time \(t\) the physically admissible states form a Borel set \(\A_t\subseteq\X\), through which hard boundaries are represented even when $\X$ has none. Before collisions or other hybrid events the deterministic dynamics are a \(C^2\) flow
\begin{equation}
 \Phi_{t,s}:\X\longrightarrow\X,\qquad s\ge t,
 \label{eq:flow}
\end{equation}
with \(\Phi_{s,r}\circ\Phi_{t,s}=\Phi_{t,r}\). In stochastic systems, \(\Phi\) is replaced by a Markov transition kernel \(K_{t,s}(x,\dd x')\). Hybrid dynamics may be represented on an augmented state space with discrete modes and transition maps.

Let \(Y_{0:t}\) denote the observations available through time \(t\). Bayesian inference produces a posterior probability measure
\begin{equation}
 \mu_t(B)=\Pp(x_t\in B\mid Y_{0:t}),
 \qquad B\in\mathcal B(\X),
 \label{eq:posterior}
\end{equation}
with \(\mu_t(\A_t)=1\). A future experiment over the horizon \([t,t+T]\) is a family of probability laws
\begin{equation}
 x\longmapsto P_{t,x}^{T}
 \label{eq:futurelaw}
\end{equation}
on a measurable output space \((\Omega_T,\mathcal F_T)\): a path of observations, a target-plane coordinate, a conjunction observable, a sequence of detections or another prediction product. An admissible experiment is fixed independently of the desired conclusion by the measurement protocol (sensor map, cadence, bandwidth, noise law, censoring, released data product). One may not add a target indicator merely to make a factorization hold unless it is actually observed; conversely, if the protocol supplies a full trajectory, any deterministic label measurable from it is derived from the experiment. This prevents both artificial impoverishment and post-hoc enlargement of the output space.

\subsection{Reference measures and the reference-relative density}

A density is not a scalar under a nonlinear coordinate change, so a highest-density region is ambiguous unless the volume convention is transformed together with the probability measure. We introduce a positive reference measure \(\nu_t\), locally equivalent to Lebesgue measure in every chart, and require \(\mu_t\ll\nu_t\). The construction is covariant under simultaneous pushforward, but is \emph{not} invariant under replacing $\nu_t$ by a different physical reference measure: such a replacement can change density rankings and topology. The \emph{reference-relative posterior density} is
\begin{equation}
 r_t=\frac{\dd\mu_t}{\dd\nu_t},
 \qquad
 f_t=-\log r_t,
 \label{eq:score-density}
\end{equation}
and $f_t$ is called the \emph{log-density index}. The word \emph{score} is reserved throughout for the derivative of a log likelihood, as in Eq.~\eqref{eq:ordinary-score}; the two objects are unrelated, and the earlier term ``posterior score density'' is not used. For Hamiltonian orbital dynamics, a natural choice of $\nu_t$ is Liouville volume restricted to the admissible phase space; other choices are allowed, but must be declared and transported under coordinate transformations. A prior measure can also serve as \(\nu_t\), in which case \(r_t\) is a posterior-to-prior density ratio.

\begin{assumption}[Regular foundational regime]
\label{ass:regular}
Unless stated otherwise, we assume:
\begin{enumerate}[label=(A\arabic*)]
\item \(\mu_t\ll\nu_t\), \(r_t\) has a continuous representative on \(\A_t\), and \(\mu_t(\A_t)=1\);
\item the selected continuum mass levels have no density plateau, that is \(\mu_t(\{r_t=c\})=0\) at their thresholds;
\item the future-law family \(x\mapsto P_{t,x}^T\) is differentiable in quadratic mean;
\item whenever smooth transport results are invoked, \(\Phi_{t,s}\) is a diffeomorphism on an open neighborhood of the relevant credible set;
\item when a continuum barcode is invoked, the corresponding module over the declared regular mass interval is q-tame over a fixed coefficient field \(\F\); when only a finite mass grid is computed, the claim is restricted to that finite persistence module.
\end{enumerate}
\end{assumption}

\begin{remark}[Scope of the smooth regime and hard boundaries]
\label{rem:hard-boundaries}
Hard admissibility constraints are encoded in $\A_t$ and need not have a globally smooth boundary; Assumption~\ref{ass:regular} is invoked componentwise, transversality entering only through Theorem~\ref{thm:isotopy} and the no-plateau condition only at selected thresholds. A finite grid or weighted particle approximation, such as the short-arc benchmark of Section~\ref{sec:numerical}, is a numerical representative rather than an instance of the smooth-density regime: its topology must be assessed under mass-level and mesh refinement, and Theorem~\ref{thm:isotopy} does not certify boundary-touching grid levels.
\end{remark}

Plateaus, singular posteriors and noninvertible dynamics require set-valued thresholds, stratified measures or hybrid transition operators and are not hidden inside the smooth notation below.

Unless a section explicitly suppresses the estimation epoch, credible sets are written $C_t^m$, transport classes $K_{t,T}^{\lambda}$ and conditioned sets $\Pi_{t,T}^{m,\lambda}$. The short-arc example sets $t=0$ and writes $C_0^m$; the autonomous PCR3BP section suppresses $t$ after declaring the fixed section.

Table~\ref{tab:notation} collects the recurring symbols; it is placed here rather than in an appendix because the notation is used from Section~\ref{sec:definition} onward.

\begin{table*}[t]
\centering
\small
\caption{\footnotesize Principal notation.}
\label{tab:notation}
\begin{tabularx}{0.95\textwidth}{@{}p{0.20\textwidth}X@{}}
\toprule
Symbol & Meaning \\
\midrule
\(\X,\ \A_t\) & state manifold; physically admissible state set \\
\(\mu_t,\ \nu_t\) & posterior probability measure; declared reference measure \\
\(r_t=\dd\mu_t/\dd\nu_t\) & reference-relative posterior density \\
\(f_t=-\log r_t\) & log-density index; $\rho_f$ is its mass rank \\
\(C_t^m,\ \Cfam_t\) & credible set of enclosed mass \(m\); mass-indexed credible filtration \\
\(P_{t,x}^T\) & law of the future output over the horizon \([t,t+T]\), given the state \(x\) \\
\(g_{t,T}\) & finite-horizon Fisher pseudometric, the pullback of \(P_{t,x}^T\) \\
\(\Hfam_t\) & persistence module of the credible filtration \\
\(\Kfam_{t,T},\ \Lambda\) & labelled transport skeleton; its label set \\
\(\Pfam_{t,T}\) & transport-conditioned persistence modules \\
\(E_{t,T}(H),\ \Pi_{t,T}^m(H)\) & states hitting \(H\) by \(T\); its mass-indexed conditioned filtration \\
\(\Gdug_t^T\) & the DUG report of Eq.~\eqref{eq:dug-operational} \\
\(\varepsilon_f,\ \varepsilon_\mu\) & uniform log-density error; total-variation error \\
\(\omega_f,\ b,\ b_h\) & index-concentration modulus; continuum and discrete mass budgets \\
\(U_h,\ D_h\) & unresolved and directly disagreeing nested-grid mass \\
\bottomrule
\end{tabularx}
\end{table*}

\section{Mass-ranked credible geometry and the DUG interface}
\label{sec:definition}

\subsection{Mass-indexed credible filtration}

Let
\begin{equation}
 G_t(c)=\mu_t\{x:r_t(x)\ge c\}.
 \label{eq:enclosed-mass-map}
\end{equation}
Call $m$ a \emph{regular mass} when $m=G_t(c)$ for a threshold $c$ that is a regular value of $r_t$ on the relevant strata with $\mu_t\{r_t=c\}=0$; this image-based definition avoids assuming every $m\in(0,1)$ is attained. If the law of $r_t$ is atomless then $G_t$ is continuous and every interior mass is attained.

For \(m\in(0,1)\), define the generalized threshold
\begin{equation}
 c_t(m)=\sup\left\{c\ge0:\mu_t\bigl(\{x\in\A_t:r_t(x)\ge c\}\bigr)\ge m\right\}.
 \label{eq:threshold}
\end{equation}
At a regular mass level, the \emph{credible admissible set} is
\begin{equation}
 C_t^m=\{x\in\A_t:r_t(x)\ge c_t(m)\},
 \qquad \mu_t(C_t^m)=m.
 \label{eq:credible}
\end{equation}
The family \(\Cfam_t=\{C_t^m\}_{m\in\mathcal M_t}\) over a chosen set of regular masses is the \emph{credible filtration}; indexing it by enclosed mass rather than by an arbitrary density value is what makes comparisons across times and numerical representations meaningful.

\begin{proposition}[Filtration property]
\label{prop:filtration}
Under Assumption~\ref{ass:regular}, if \(m_1<m_2\) then \(C_t^{m_1}\subseteq C_t^{m_2}\) up to a \(\mu_t\)-null set. Moreover, \(\mu_t(C_t^m)=m\) at every regular mass in the image of Eq.~\eqref{eq:enclosed-mass-map}.
\end{proposition}
\begin{proof}
The map $G_t$ is nonincreasing. Hence \(m_1<m_2\) implies \(c_t(m_1)\ge c_t(m_2)\), which gives the inclusion of superlevel sets. By the definition of a regular mass there is a selected threshold $c$ with $G_t(c)=m$ and no mass on its level set. Monotonicity and the generalized-inverse definition then give $c_t(m)=c$ up to a score interval carrying no mass, so $\mu_t(C_t^m)=G_t(c)=m$.
\end{proof}

\begin{remark}[Plateaus]
\label{rem:plateaus}
If \(\mu_t(\{r_t=c_t(m)\})>0\), exact mass can be achieved by adding a measurable subset of the plateau. The resulting credible set is not canonical. In that regime the DUG should retain the interval of admissible sets or a tie-breaking rule, and topological conclusions that depend on the selected subset must be reported as nonunique.
\end{remark}

\subsection{Finite-horizon Fisher pseudometric}

Differentiability in quadratic mean means that, in a chart and for every tangent direction \(u\in T_x\X\), there is a score \(\ell_{t,T,x}(u)\in L^2(P_{t,x}^T)\), linear in \(u\), such that the square-root density has first variation \(\tfrac12\ell_{t,T,x}(u)\sqrt{p_{t,x}^T}\) \cite{VanderVaart1998,Amari2000}. We define
\begin{equation}
 g_{t,T,x}(u,v)
 =\E_{P_{t,x}^{T}}
 \left[\ell_{t,T,x}(u)\,\ell_{t,T,x}(v)\right].
 \label{eq:fisher}
\end{equation}
Under common domination, differentiability in $x$ and a local integrable envelope, the ordinary score is
\begin{equation}
 \ell_{t,T,x}(u)=\partial_u\log p_{t,x}^T,
 \label{eq:ordinary-score}
\end{equation}
and \eqref{eq:fisher} is the Fisher pullback tensor. We use the term \emph{pseudometric} because physically relevant directions can be unobservable and therefore have zero length.

\begin{proposition}[Identifiability meaning of the metric]
\label{prop:pseudometric}
For each \(x\), \(g_{t,T,x}\) is a symmetric positive-semidefinite bilinear form. Its null space is
\begin{equation}
 \mathcal N_{t,T,x}
 =\{u\in T_x\X:\ell_{t,T,x}(u)=0\ \ P_{t,x}^{T}\text{-a.s.}\}.
 \label{eq:nullspace}
\end{equation}
Thus \(u\in\mathcal N_{t,T,x}\) precisely when the future experiment is insensitive to the perturbation \(u\) to first order.
\end{proposition}
\begin{proof}
Bilinearity and symmetry follow from linearity of the score and the \(L^2\) inner product. Positive semidefiniteness follows from
\(g_{t,T,x}(u,u)=\|\ell_{t,T,x}(u)\|_{L^2(P_{t,x}^T)}^2\ge0\). Equality holds exactly when the score vanishes almost surely.
\end{proof}

For deterministic dynamics, a mathematically well-defined continuous-time Gaussian experiment is
\begin{equation}
 \dd Z_s=h_s(\Phi_{t,s}(x))\,\dd s+R_s^{1/2}\,\dd W_s,
 \qquad s\in[t,t+T],
\label{eq:gaussian-observation}
\end{equation}
where \(R_s\) is positive definite. Its Fisher pullback has the computable form
\begin{equation}
 g_{t,T,x}
 =\int_t^{t+T}
 D(h_s\circ\Phi_{t,s})(x)^{\T}R_s^{-1}
 D(h_s\circ\Phi_{t,s})(x)\,\dd s.
 \label{eq:gramian}
\end{equation}
This is a finite-horizon nonlinear observability Gramian, related to differential observability constructions \cite{HermannKrener1977}; it measures distinguishability in the chosen experiment, not sensitivity of the flow. For discrete epochs the integral becomes a sum. The metric is local:
\begin{equation}
 D_{\mathrm{KL}}\!\left(P_{t,x}^T\,\middle\|\,P_{t,x+\delta u}^T\right)
 =\frac{\delta^2}{2}g_{t,T,x}(u,u)+o(\delta^2).
 \label{eq:local-kl}
\end{equation}
Finite separation of distant posterior components is therefore a property of the future-law family, not of the Fisher tensor alone.

\subsection{Persistent topology of the credible filtration}

Fix one coefficient field \(\F\) for every module being compared; continuum groups mean singular homology, which on compact triangulable sets agrees with the simplicial computation. For \(m_1<m_2\) the inclusion \(C_t^{m_1}\hookrightarrow C_t^{m_2}\) induces
\begin{equation}
 \iota_{m_1,m_2}^{(k)}:
 H_k(C_t^{m_1};\F)\longrightarrow H_k(C_t^{m_2};\F).
 \label{eq:persistence-map}
\end{equation}
The collection
\begin{equation}
 \Hfam_t^{(k)}=
 \left(\{H_k(C_t^m;\F)\}_{m\in\mathcal M_t},
       \{\iota_{m_1,m_2}^{(k)}\}_{m_1<m_2}\right)
 \label{eq:persistence-module}
\end{equation}
is the degree-\(k\) persistence module of the credible filtration; a q-tame continuum module admits a barcode \(\Dgm_k(\Cfam_t)\) \cite{Edelsbrunner2002,Carlsson2009,Chazal2016}, whereas on a finite mass grid the output is the finite restricted module and its grid-dependent decomposition, which is not evidence for unsampled levels. When boundary strata or singular subsets matter, a compatible Whitney stratification is declared as additional input. Using the whole filtration avoids attaching physical significance to a single confidence level: a component that persists across a broad mass range is distinguished from one created by a narrow threshold interval or by numerical noise.

\subsection{Mass-rank stability: main theorem}
\label{sec:stability}

Standard stability controls filtrations indexed by a fixed function, typically
the raw index $f=-\log r$, under perturbation of that function
\cite{CohenSteiner2007,Chazal2016}. Here the filtration is indexed instead by
enclosed posterior mass, so that the index is itself a functional of the
measure: under numerical approximation both $f$ and the law of $f$ move, and
the reparameterization is part of the perturbation. For a positive density
$r=\dd\mu/\dd\nu$ on a common compact domain, define
\begin{equation}
 f=-\log r,
 \qquad F_f(c)=\mu\{f\le c\},
 \qquad \rho_f=F_f\circ f,
 \label{eq:mass-rank}
\end{equation}
and the two measure-specific index-concentration moduli
\[
 \omega_f(a)=\sup_{c\in\R}\mu\{y:|f(y)-c|\le a\}
 \]
\begin{equation}
 \omega_{\widehat f}(a)=\sup_{c\in\R}\widehat\mu
 \{y:|\widehat f(y)-c|\le a\}.
 \label{eq:score-modulus}
\end{equation}
The rank sublevel sets $C_f^m=\{\rho_f\le m\}$ are defined even with plateaux. When the law of the index is atomless they agree with the exact-mass credible filtration under the convention of Eq.~\eqref{eq:credible}; with a plateau they retain the complete tied level and can overshoot $m$, as discussed in Remark~\ref{rem:plateaus}.
Notice that $\omega_f(0)=\sup_c\mu\{f=c\}$ can be positive. Therefore the bound below need not vanish for atomic indices or continuum plateaux; this is precisely the content excluded by Assumption~\ref{ass:regular}(A2).

Two error levels enter, and they are kept typographically distinct: $\varepsilon_f$ bounds the uniform log-density error and $\varepsilon_\mu$ bounds the total-variation error of the measure.

\begin{theorem}[Mass-rank interleaving and persistence stability]
\label{thm:persistence-stability}
Let $\mu$ and $\widehat\mu$ have positive continuous densities $r$ and
$\widehat r$ relative to the same declared reference measure on a compact domain $X$.
Put $f=-\log r$, $\widehat f=-\log\widehat r$ and
$\widehat\rho_{\widehat f}(x)=\widehat\mu\{\widehat f\le\widehat f(x)\}$.
If
\begin{equation}
 \|f-\widehat f\|_\infty\le\varepsilon_f,
 \qquad
 d_{\TV}(\mu,\widehat\mu)\le\varepsilon_\mu,
 \label{eq:mass-rank-assumptions}
\end{equation}
then, with
\begin{equation}
 b=\varepsilon_\mu+\min\{\omega_f(\varepsilon_f),
                       \omega_{\widehat f}(\varepsilon_f)\},
 \label{eq:mass-rank-budget}
\end{equation}
\begin{equation}
 \|\rho_f-\widehat\rho_{\widehat f}\|_\infty\le b.
 \label{eq:mass-rank-bound}
\end{equation}
Extend the mass filtrations by $C^m=\varnothing$ for $m<0$ and $C^m=X$
for $m>1$. Then
\begin{equation}
 C_f^{m-b}\subseteq C_{\widehat f}^{m}
 \subseteq C_f^{m+b},
 \qquad
 C_{\widehat f}^{m-b}\subseteq C_f^{m}
 \subseteq C_{\widehat f}^{m+b}.
 \label{eq:mass-rank-interleaving}
\end{equation}
Thus their homology functors $\R\to\mathrm{Vect}_{\F}$ are $b$-interleaved. If they are q-tame, then for every degree $k$,
\begin{equation}
 d_I\!\left(\Hfam_f^{(k)},\Hfam_{\widehat f}^{(k)}\right)\le b,
 \qquad
 d_B\!\left(\Dgm_k(\Cfam_f),\Dgm_k(\Cfam_{\widehat f})\right)\le b,
 \label{eq:persistence-stability}
\end{equation}
with all modules and diagrams parameterized by enclosed mass.
\end{theorem}
\begin{proof}
The two one-sided comparisons, the independent bounds using each concentration modulus, endpoint extensions and persistence-category step are proved line by line in Appendix~\ref{app:mass-rank-proof}.
\end{proof}

\begin{corollary}[Quantitative refinement consistency]
\label{cor:mass-rank-consistency}
If the push-forward $f_\#\mu$ has a Lebesgue density bounded by $L$, then
\begin{equation}
 \omega_f(\varepsilon_f)\le 2L\varepsilon_f.
 \label{eq:bounded-score-density}
\end{equation}
Consequently $b\le\varepsilon_\mu+2L\varepsilon_f$ whenever this modulus is the smaller one in Eq.~\eqref{eq:mass-rank-budget}. For an approximation sequence with $\varepsilon_{\mu,h}\to0$, $\varepsilon_{f,h}\to0$ and a uniform bound $L$, the mass-rank interleaving and bottleneck bound converge to zero.
\end{corollary}
\begin{proof}
Every window $[c-\varepsilon_f,c+\varepsilon_f]$ has length $2\varepsilon_f$, so its push-forward probability is at most $2L\varepsilon_f$.
\end{proof}

\begin{remark}[When the bound is informative, and how to estimate it]
\label{rem:informative}
Validity is not the same as usefulness: a budget $b$ larger than the mass persistence of the features of interest cannot select a resolution. Equation~\eqref{eq:mass-rank-budget} is informative exactly when the log-density index does not concentrate. Three sufficient regimes are worth naming. (i) If $f_\#\mu$ has a density bounded by $L$, Corollary~\ref{cor:mass-rank-consistency} gives $b\le\varepsilon_\mu+2L\varepsilon_f$, so a first-order refinement in the density gives a first-order budget. (ii) On a compact domain with $\|\nabla f\|\ge\kappa>0$ outside a $\mu$-null critical set and bounded reference volume, the coarea formula bounds $L$ by $\nu(X)\sup r/\kappa$. (iii) If the index has a plateau of mass $p$, then $\omega_f(0)\ge p$ and no refinement can bring $b$ below $p$; this is the failure mode that Assumption~\ref{ass:regular}(A2) excludes and that atomic grids always exhibit. In practice $\omega_f$ need not be bounded analytically: it is the supremum of a sliding-window probability and is computed exactly on a weighted grid by Eq.~\eqref{eq:empirical-omega}, so a practitioner can read the concentration penalty off the data before deciding whether the budget is usable. Table~\ref{tab:mass-rank-check} reports two cases in which the penalty dominates.
\end{remark}

\begin{corollary}[Finite atomic mass-rank theorem]
\label{cor:atomic-mass-rank}
Let $X=\{x_1,\ldots,x_N\}$ carry a common positive atomic reference measure $\nu=\sum_i v_i\delta_{x_i}$, and let $\mu_i=r_i v_i$ and $\widehat\mu_i=\widehat r_i v_i$ be positive probability weights. If the two errors in Eq.~\eqref{eq:mass-rank-assumptions} are bounded by $\varepsilon_f$ and $\varepsilon_\mu$, then Eqs.~\eqref{eq:mass-rank-bound}--\eqref{eq:persistence-stability} hold for the finite rank filtrations $C_f^m=\{x_i:\rho_f(x_i)\le m\}$ with the discrete moduli in Eq.~\eqref{eq:score-modulus}. No atomlessness assumption is needed. A complete-tie superlevel mask is assigned its realized cumulative mass, at which it is exactly a rank-filtration set.
\end{corollary}
\begin{proof}
The pointwise set inclusions in the proof of Theorem~\ref{thm:persistence-stability} use only the uniform index error and total variation, not continuity or atomlessness. A finite filtration is q-tame, so the interleaving and bottleneck conclusions follow directly.
\end{proof}

Although the index comparison produces endpoints separated by $2\varepsilon_f$, the added or removed interval itself has length $2\varepsilon_f$ and is contained in a radius-$\varepsilon_f$ window centered at a shifted index value; this is why Eq.~\eqref{eq:mass-rank-budget} contains $\omega(\varepsilon_f)$ rather than $\omega(2\varepsilon_f)$.

The continuum theorem and Corollary~\ref{cor:mass-rank-consistency} control a refinement limit when index distributions do not concentrate on plateaux. Corollary~\ref{cor:atomic-mass-rank} gives a rigorous finite-grid statement: the short-arc quantity $b_h$ in Section~\ref{sec:representation-diagnostics} certifies the two rank filtrations on their declared common validation grid. It does not by itself certify convergence to a continuum filtration; that separate claim still requires interpolation, quadrature and index-CDF envelopes uniform in $h$.

\paragraph{A continuous semi-analytical example.}
To show the first role directly, let $X=[-4,4]$ with Lebesgue reference measure and
\[
 r(x)=Z_0^{-1}\left\{\tfrac12\phi_{0.6}(x-1.2)
                +\tfrac12\phi_{0.6}(x+1.2)\right\}
 \]               
 \begin{equation}
 \widehat r(x)=M_\eta^{-1}r(x)e^{\eta x},\quad \eta=0.02,
 \label{eq:continuous-example}
\end{equation}
where $Z_0$ truncates the mixture to $X$, $M_\eta=\int_Xr(x)e^{\eta x}\,\dd x$ and $\phi_\sigma$ is a centered Gaussian density. Both densities are smooth, positive and bimodal. Their score difference is affine,
\begin{equation}
 \widehat f(x)-f(x)=-\eta x+\log M_\eta,
 \label{eq:continuous-score-difference}
\end{equation}
so $\varepsilon_f$ is attained at an endpoint; $\varepsilon_\mu=\tfrac12\int_X|r-\widehat r|\,\dd x$. The index-window suprema, mass ranks and critical-level masses reduce to one-dimensional integration. Deterministic midpoint quadrature with $4\times10^5$ cells gives Table~\ref{tab:continuous-mass-rank}; doubling the grid changes every displayed quantity by less than $10^{-5}$.

\begin{table*}[t]
\vspace*{1.5em}
\centering
\small
\caption{\footnotesize Direct continuous example for Theorem~\ref{thm:persistence-stability}. The only finite $H_0$ bar is reported as enclosed-mass birth/death coordinates.}
\label{tab:continuous-mass-rank}
\begin{tabular}{@{}lrr@{}}
\toprule
Quantity & $r$ & $\widehat r$ \\
\midrule
$\varepsilon_f=\|f-\widehat f\|_\infty$ & \multicolumn{2}{c}{0.08036} \\
$\varepsilon_\mu=d_{\mathrm{TV}}(\mu,\widehat\mu)$ & \multicolumn{2}{c}{0.01210} \\
Index concentration $\omega(\varepsilon_f)$ & 0.43115 & 0.39962 \\
Finite $H_0$ bar $(m_{\rm birth},m_{\rm death})$ & $(0,0.94703)$ & $(0.12494,0.94705)$ \\
Bound $b$ & \multicolumn{2}{c}{0.41172} \\
Observed $\|\rho_f-\widehat\rho_{\widehat f}\|_\infty$ & \multicolumn{2}{c}{0.12973} \\
$H_0$ bottleneck distance & \multicolumn{2}{c}{0.12494} \\
\bottomrule
\end{tabular}
\end{table*}

The bound is deliberately worst-case: total variation contributes only $0.01210$, whereas the selected index window can contain about $0.40$ mass. That concentration penalty is paid uniformly even though the actual rank displacement is $0.12973$ and the only finite bar moves by $0.12494$. The gap therefore quantifies conservatism caused mainly by threshold concentration, rather than failure of the theorem or of the mass coordinate.

\subsection{Dynamical skeleton and transport-conditioned filtrations}

For each label $\lambda$, let $Z_\lambda$ be a second-countable Hausdorff descriptor space with Borel $\sigma$-algebra $\mathcal Z_\lambda$. A \emph{model-induced dynamical descriptor} is a declared Borel map
\begin{equation}
 \Gamma_{t,T}^{\lambda}:(\A_t,\mathcal B(\A_t))
 \longrightarrow (Z_\lambda,\mathcal Z_\lambda)
 \label{eq:descriptor-map}
\end{equation}
constructed from the flow, variational flow, event map, committor, resonance condition or coherent-structure diagnostic. Each label has a declared Borel target $B_\lambda\in\mathcal Z_\lambda$ and defines
\begin{equation}
 K_{t,T}^{\lambda}
 =\left(\Gamma_{t,T}^{\lambda}\right)^{-1}(B_\lambda)\cap\A_t.
 \label{eq:skeleton-element}
\end{equation}
These are the minimal conditions needed for posterior probabilities and intersections to be well defined. Continuum homology is reported only when $C_t^m$ is compact and each conditioned set is compact and tame, for instance triangulable, a compact ENR or definable in a fixed o-minimal structure; Borel first-hit classes failing such conditions still define probabilities, but no continuum persistence module is claimed for them. Under a coordinate change $\psi$ the represented descriptor is $\Gamma_{t,T}^{\lambda}\circ\psi^{-1}$, so $K_{t,T}^{\prime\lambda}=\psi(K_{t,T}^{\lambda})$; a descriptor failing measurability, declared regularity or this covariance rule is not an admissible skeleton component.

The \emph{dynamical skeleton}
\begin{equation}
 \Kfam_{t,T}=
 \left(\{K_{t,T}^{\lambda}\}_{\lambda\in\Lambda},
       \mathcal I,\mathcal L\right)
 \label{eq:skeleton}
\end{equation}
consists of a locally finite labelled family, an incidence relation \(\mathcal I\) and the level-incidence data \(\mathcal L(\lambda,m)=K_{t,T}^{\lambda}\cap C_t^m\). Examples include resonance zones, invariant manifolds, scattering channels, target-plane keyholes, coherent sets and collision surfaces \cite{Wiggins1992,Shadden2005,Haller2015}. Requiring a \emph{declared} descriptor prevents arbitrary clustering in a learned latent space from being called a dynamical skeleton. Both numerical examples use a finite label set.

Three terms are used separately: a \emph{trajectory} is a realization $s\mapsto\Phi_{t,s}(x)$; a \emph{transport class} $K_{t,T}^{\lambda}$ is the set of initial conditions carrying one declared label; a \emph{channel component} is a connected component of $C_t^m\cap K_{t,T}^{\lambda}$. The informal word \emph{pathway} is reserved for a physical mechanism and is never used as a synonym for a label or a component.

For a closed target or hazard set \(H\subseteq\X\), define the hitting time in elapsed time $s$ after the estimation epoch $t$:
\begin{equation}
 \tau_H(x)=\inf\{s\in[0,T]:\Phi_{t,t+s}(x)\in H\},
 \qquad \inf\varnothing=+\infty,
 \label{eq:hitting-time}
\end{equation}
and the event-generating set
\begin{equation}
 E_{t,T}(H)=\{x\in\A_t:\tau_H(x)\le T\}.
 \label{eq:event-set}
\end{equation}
If the flow is continuous on $C_t^m\times[0,T]$ then $E_{t,T}(H)\cap C_t^m$ is compact. With several targets a first-hit label uses the least localized hitting time; simultaneous hits follow a declared deterministic priority rule and their posterior mass is reported separately. Such tie breaking generally produces Borel, possibly half-open classes, so a compact tame representative must be declared before continuum homology is attached. The PCR3BP calculation uses this convention, finds zero ambiguous-event mass at its localization tolerance, and reports finite-grid class incidence only.
The \emph{event-conditioned transport filtration} is
\begin{equation}
 \Pi_{t,T}^{m}(H)=C_t^m\cap E_{t,T}(H),
 \qquad m\in\mathcal M_t,
 \label{eq:pathway-filtration}
\end{equation}
with event probability
\begin{equation}
 p_{t,T}(H)=\mu_t(E_{t,T}(H)).
 \label{eq:event-probability}
\end{equation}
For each degree \(k\), its inclusions define the \emph{transport-conditioned persistence module}
\begin{equation}
 \Pfam_{t,T}^{(k)}(H)=
 \left(
 \{H_k(\Pi_{t,T}^{m}(H);\F)\}_{m\in\mathcal M_t},
 \{j_{m_1,m_2}^{(k)}\}_{m_1<m_2}
 \right).
 \label{eq:pathway-persistence}
\end{equation}
Every labelled skeleton element likewise has the filtration \(\Pi_{t,T}^{m,\lambda}=C_t^m\cap K_{t,T}^{\lambda}\) and, when tame, a module \(\Pfam_{t,T}^{(k)}(\lambda)\). Thus \(\Hfam_t\) records the topology of credible sets and \(\Pfam_{t,T}\) that of transport-conditioned subsets; they coincide only in special cases. For a Markov process the finite-horizon committor
\begin{equation}
 q_{t,T}^{H}(x)=\Pp_x(\tau_H\le T)
 \label{eq:committor}
\end{equation}
replaces the sharp deterministic preimage. For a declared threshold $q_0$ the conditioned sets \(C_t^m\cap\{q_{t,T}^{H}\ge q_0\}\) form a one-parameter mass filtration; varying both $m$ and $q_0$ produces a bifiltration requiring multiparameter summaries, and no one-parameter barcode is then implied.

\subsection{Operational interface and task check}
\label{sec:operational-interface}

For computation, DUG is the report
\begin{strip}
\begin{equation}
 \Gdug_t^T=
 \left(
   \underbrace{\nu_t,\mu_t,\{C_t^m,H_k(C_t^m)\}_m}_{\text{posterior mass geometry}};
   \underbrace{\{P_{t,x}^T\}_x,g_{t,T}}_{\text{future experiment}};
   \underbrace{\{K_{t,T}^{\lambda},H_k(C_t^m\cap K_{t,T}^{\lambda})\}_{m,\lambda}}_{\text{labelled transport}}
 \right).
 \label{eq:dug-operational}
\end{equation}
\end{strip}
The semicolons record provenance, not independence: the same dynamics may generate the future laws and the transport labels, and under a diffeomorphism every measure, set and tensor is transported together, so the represented report is the same physical object. Covariance checks are collected in Appendix~\ref{app:supporting-results}. A reduced report is adequate for a task $\Theta$ only when $\Theta$ factors through the retained entries on the declared model class: event probability needs posterior weights and the event set but no Fisher tensor, whereas finite label-discrimination utility needs label weights and conditional future laws, which a local quadratic tensor cannot supply.

\begin{proposition}[Finite label utility and the Fisher limitation]
\label{prop:mi-minimality}
On the class of positive weighted finite state spaces with a binary label $\Lambda$ and output $Z$, the value $I(\Lambda;Z)$ requires the state weights, label map and conditional future-law family. Retaining those weights and labels together with only the local Fisher tensor is not sufficient.
\end{proposition}
\begin{proof}
The factorization and three witness pairs are given in Appendix~\ref{app:task-witnesses}. In particular, locally constant state-independent and label-aligned future laws have the same zero local Fisher tensor but mutual information zero and one bit, respectively.
\end{proof}

\begin{corollary}[A local Fisher tensor is not task-sufficient for finite label utility]
\label{cor:fisher-not-sufficient}
Two experiment designs can have the same retained posterior weights, labels and local Fisher tensor but different finite-separation label information. Therefore optimization of $g_{t,T}$ need not optimize $I(\Lambda;Z)$.
\end{corollary}

\begin{proposition}[A non-injective observable need not determine transport]
\label{thm:skeleton-obstruction}
There are smooth deterministic flows, in the sense of Eq.~\eqref{eq:flow}, with the same posterior and the same declared non-injective future observable, for which one common credible set has different numbers of target-preimage components. A target label or an equivalent event-complete observation is therefore necessary for universal channel reconstruction.
\end{proposition}
\begin{proof}
Appendix~\ref{app:task-witnesses} gives an explicit pair of smooth nonautonomous vector fields, with their flows written in closed form, whose reported terminal $x$ coordinate agrees while a target-line preimage has one versus two components.
\end{proof}

\subsection{Update, forecast and topological transition}
\label{sec:topology-change}

Forecast and update play different geometric roles, and only the first transports the construction. A deterministic forecast pushes $\mu_t$ forward by the flow, so the credible filtration and its homology are carried along; an observational update multiplies the forecast density by a likelihood and renormalizes, so the filtration must be recomputed. There is no meaningful operation that ``pushes forward the entire report'' through an assimilation.

At fixed enclosed mass the topology of $C_t^m$ is rigid away from identifiable events. Theorem~\ref{thm:isotopy} in Appendix~\ref{app:isotopy} records the precise statement: for a $C^2$ family with a nonvanishing gradient on the moving level set and transversal intersection with $\partial\A$, the credible sets are mutually ambiently isotopic and their Betti numbers are constant. Consequently, within a smooth Bayesian family, topology at fixed mass can change only if a threshold crosses a critical point of the reference-relative density, transversality with an admissibility boundary fails, the support or admissible set changes, or the model ceases to be smooth. Projection, marginalization and finite numerical resolution can additionally create \emph{apparent} transitions, because they replace the underlying set by a different object. Large stretching and folding alone are never sufficient. The nonsmooth, cut-defined admissible regions used in orbit determination do not satisfy the hypotheses automatically: the theorem indicates which mechanisms to test, and does not certify grid intersections.

\section{Why covariance does not determine global transport}
\label{sec:gaussian-limit}

\subsection{A covariance obstruction for DUG topology}

The local Gaussian limit of Appendix~\ref{app:supporting-results} explains when covariance suffices; the following shows why no covariance-based extension recovers the global topological component in general.

\begin{theorem}[First two moments do not determine credible-set topology]
\label{thm:covariance-obstruction}
For every dimension $n\ge1$, there exist two probability measures on $\R^n$ with smooth, everywhere positive densities, identical mean and identical positive-definite covariance, and a regular mass $m\in(0,1)$ for which their highest-density credible sets have different zeroth Betti numbers. Consequently, credible-set or event-conditioned channel multiplicity is not a function of the mean and covariance.
\end{theorem}
\begin{proof}
The witnesses are $\mu_1=\mathcal N(0,I_n)$ and the equal-weight mixture
$\mu_2=\tfrac12\mathcal N(ae_1,\Sigma_a)+\tfrac12\mathcal N(-ae_1,\Sigma_a)$ with
$\Sigma_a=\operatorname{diag}(1-a^2,1,\ldots,1)$ and $a>2^{-1/2}$, which share
mean $0$ and covariance $I_n$ while having one and two superlevel components
respectively. Details are given in Appendix~\ref{app:covariance-witness}.
\end{proof}

\begin{corollary}[Covariance cannot certify channel uniqueness]
\label{cor:covariance-pathways}
Even exact knowledge of posterior mean and covariance cannot certify that a credible event preimage contains one channel component rather than several. Such certification requires the credible filtration, the event map and their persistent $H_0$ incidence.
\end{corollary}

This is not a criticism of covariance in the linear-Gaussian regime but an information-theoretic obstruction: no manipulation of two moments reconstructs an invariant they do not determine.

\section{Representation diagnostics}
\label{sec:representation-diagnostics}

DUG is coordinate-covariant, but computations use particles, meshes, mixtures, flows or graph complexes, so numerical equivalence must be defined by controlled errors in the components and in the outputs.

\subsection{Event-probability stability}

For probability measures $\mu,\widehat\mu$ and measurable event sets $E,\widehat E$, inserting and subtracting $\widehat\mu(E)$ gives the standard error decomposition
\[
 |\mu(E)-\widehat\mu(\widehat E)|
 \le d_{\TV}(\mu,\widehat\mu)
 +\widehat\mu(E\triangle\widehat E)
 \]
 \begin{equation}
 d_{\TV}(\mu,\widehat\mu)=\sup_B|\mu(B)-\widehat\mu(B)|.
 \label{eq:event-error-bound}
\end{equation}
The first term is posterior approximation error, the second is geometric error in resolving the event-generating set; reporting only Monte Carlo uncertainty is therefore misleading when a keyhole or collision surface is under-resolved.

\subsection{A computable weighted-grid diagnostic}
Theorem~\ref{thm:persistence-stability} separates log-density error, measure error and mass concentrated near an index threshold; different domains or reference measures first require a declared identification and pullback. The following calculation borrows that decomposition for represented grids. Suppose two positive approximations are evaluated on common validation cells
with reference masses $v_i$ and normalized posterior weights $w_i$ and
$\widehat w_i$. Set
\begin{align}
 f_i&=-\log(w_i/v_i),
 &\widehat f_i&=-\log(\widehat w_i/v_i),\\
 \varepsilon_{f,h}&=\max_i|f_i-\widehat f_i|+\eta_{\mathrm{int}},
 & \\
 \varepsilon_{\mu,h}&=\frac12\sum_i|w_i-\widehat w_i|+\eta_{\mathrm{TV}},
 \label{eq:empirical-eps-delta}
\end{align}
where the envelopes include interpolation and unresolved-cell errors. Zero
weights cannot be hidden by a log floor: excluded cells must be added to
$\eta_{\mathrm{TV}}$. For either weighted index array compute the exact
sliding-window statistic
\begin{equation}
 \widehat\omega_f(a)
 =\max_{u\in\R}\sum_i w_i\,
 \mathbf 1\{u\le f_i\le u+2a\}.
 \label{eq:empirical-omega}
\end{equation}
A verified index-CDF envelope $\eta_{\mathrm{cdf}}$ gives
$\omega_f(a)\le\widehat\omega_f(a)+2\eta_{\mathrm{cdf}}$.
Use $\widehat\eta_{\mathrm{cdf}}$ analogously for the second approximation.
For $N$ independent draws, the Dvoretzky--Kiefer--Wolfowitz inequality gives
\(
\eta_{\mathrm{cdf}}=\sqrt{\log(2/\alpha)/(2N)}
\)
with confidence $1-\alpha$; self-normalized importance weights require a
separate concentration bound.

The resulting inter-resolution diagnostic is
\begin{equation}
 b_h=\varepsilon_{\mu,h}+\min\left\{
 \widehat\omega_f(\varepsilon_{f,h})+2\eta_{\mathrm{cdf}},
 \widehat\omega_{\widehat f}(\varepsilon_{f,h})
 +2\widehat\eta_{\mathrm{cdf}}
 \right\}.
 \label{eq:empirical-bottleneck}
\end{equation}
Under the stated envelopes, a feature with mass persistence greater than
$2b_h$ cannot be matched to the diagonal. Without verified interpolation,
quadrature and CDF envelopes, the calculation is an inter-resolution comparison
rather than a continuum certificate.

\begin{table*}[t]
\centering
\small
\caption{\footnotesize Weighted-grid diagnostics motivated by Theorem~\ref{thm:persistence-stability}. The short-arc row compares a bilinear interpolation of the $\grid{101}$ log profile density with direct evaluation on a common $\grid{201}$ validation grid, conditional on the union of their tie-complete 99.9\% regions. No continuum envelopes are asserted.}
\label{tab:mass-rank-check}
\begin{tabular}{@{}lrrrrr@{}}
\toprule
Case & $\varepsilon_{f,h}$ & $\varepsilon_{\mu,h}$ & $\min\omega$ & $b_h$ & Observed $\|\Delta\rho\|_\infty$ \\
\midrule
Four-cell exact example & 0.10536 & 0.03000 & 0.40000 & 0.43000 & 0.02000 \\
Short arc, $\grid{101}\to\grid{201}$ & 0.00488 & 0.00059 & 0.07336 & 0.07395 & 0.01963 \\
\bottomrule
\end{tabular}
\end{table*}

Both gaps are explained by the concentration term: in the four-cell case the bound $0.430$ is dominated by a $0.400$ index window while the realized rank change is $0.020$, and in the short-arc comparison the window contributes $0.07336$ against an observed rank error of $0.01963$. So $b_h$ is a one-sided safety budget, not an error estimator. By Corollary~\ref{cor:atomic-mass-rank} the second row is rigorous for the two declared weighted-grid filtrations on their common validation grid; it still lacks the envelopes needed for a continuum-refinement certificate.

\subsection{Exact relationship between the stability results and the benchmarks}
\label{sec:theorem-scope}

The analytical and numerical parts of this paper meet in a limited and specific
way, and it is worth stating exactly where. Table~\ref{tab:theorem-scope} lists
every computed quantity reported below together with the result, if any, that
controls it.

\begin{table*}[t]
\centering
\scriptsize
\caption{\footnotesize Which result controls which computed quantity. ``Empirical'' means that the quantity is reported as a reproducible computation with declared conventions and refinement studies, not as a certified consequence of a theorem in this paper.}
\label{tab:theorem-scope}
\begin{tabularx}{\textwidth}{@{}p{0.30\textwidth}p{0.24\textwidth}X@{}}
\toprule
Computed quantity & Controlled by & Status \\
\midrule
Continuous tilt example, Table~\ref{tab:continuous-mass-rank} & Theorem~\ref{thm:persistence-stability} & Certified: all hypotheses verified directly, up to the declared quadrature error $3\times10^{-5}$ \\
Short-arc $b_h$, Table~\ref{tab:mass-rank-check} & Corollary~\ref{cor:atomic-mass-rank} & Certified for the two declared weighted-grid filtrations on their common validation grid; not a continuum statement \\
Convergence of $b_h$ as $h\to0$ & Corollary~\ref{cor:mass-rank-consistency} & Conditional: requires verified interpolation, quadrature and index-CDF envelopes uniform in $h$, which are not supplied \\
Short-arc barcode, Section~\ref{sec:barcode} & --- & Empirical. It is an exact degree-zero barcode of the declared finite graph filtration, and Corollary~\ref{cor:atomic-mass-rank} bounds its bottleneck displacement between two grids by $b_h$; no continuum barcode is claimed \\
Short-arc class masses and 100-seed ensemble & --- & Empirical, conditional on the declared model, gate, adjacency and noise law \\
Label information and Fisher eigenvalue, Section~\ref{sec:design} & Corollary~\ref{cor:fisher-not-sufficient} & The qualitative separation is an instance of the corollary; the numerical values are empirical with reported quadrature and step-size audits \\
PCR3BP class masses and intervals & --- & Empirical; $U_h$ and $D_h$ are measured, not bounded a priori \\
$\beta_0(C^m)=1$ for the PCR3BP posterior & --- & Empirical finite-grid statement \\
\bottomrule
\end{tabularx}
\end{table*}

Two consequences deserve emphasis. First, the continuum theorem is exercised in
full only by the one-dimensional example of Eq.~\eqref{eq:continuous-example};
this is deliberate, since that is the only configuration in which every
hypothesis can be checked analytically. Second, no statement below promotes a
finite-grid count to a continuum Betti number. Where the text says that a class
``is resolved'' at a given mass, it means resolved as a component of the
declared finite graph filtration, with the refinement, adjacency and noise
studies reported alongside.

\section{Relation to persistence and transport methods}
\label{sec:classical}
\subsection{Persistence, coherent sets, committors and set-oriented numerics}
\label{sec:tda-positioning}

Persistent homology is used twice, on different filtrations: $\beta_k(C_t^m)$ describes posterior topology and $\beta_k(C_t^m\cap K_{t,T}^\lambda)$ a transport-conditioned subset, and their inclusion maps across $m$ need not agree. When the index is Morse and admissibility boundaries are regular, $\Cfam_t$ is a superlevel-set Morse filtration reparameterized by posterior mass; Reeb graphs, merge trees and Mapper approximate it subject to their usual lens, cover and scale choices \cite{Biasotti2008,Singh2007,CarriereOudot2018,MischaikowNanda2013}.

Coherent-set, set-oriented and transition-path methods solve different problems: they produce almost-invariant or invariant sets from transfer operators and box coverings \cite{DellnitzJunge2002,FroylandPadberg2009,Froyland2013,DellnitzKlusZiessler2017}, or committors, reactive densities and currents between specified sets \cite{EVandenEijnden2006,EVE2010}. Any of these can supply $K_{t,T}^\lambda$ or, through Eq.~\eqref{eq:committor}, the descriptor itself. DUG adds neither their objectives nor their solvers. It adds the declared posterior and reference measure, and records which such mechanism intersects a credible set, at what posterior mass, and whether the declared future experiment distinguishes the results. Conversely, a posterior cluster is not a coherent set merely because it persists across density levels.

Two coupled statements follow. ``Both $L_1$ and $L_2$ labels occur throughout a declared mass range'' is implied neither by the posterior filtration, which carries no exit label, nor by a transport partition, which carries no inference-time mass: it is a statement about $\Pfam_{t,T}$. And ``a future observation discriminates posterior-weighted transport labels'' requires $\Lambda$, its weights and the conditional densities $p(z\mid\Lambda)$; neither a committor nor an unlabelled future-law family defines $I(\Lambda;Z)$.

\section{Minimum computational report}
\label{sec:computation}

A numerical DUG is a report assembled from existing tools, not a prescribed pipeline. At minimum it should declare $\A_t$, $\nu_t$, the calibration status of $\widehat\mu_t$ and the mass grid; declare the future-law family and whether it contains every event functional used later; declare each descriptor, its label and its numerical incidence rule; compute only the views the scientific output uses; and refine posterior, index, topology and label boundaries separately, reporting Eqs.~\eqref{eq:event-error-bound}, \eqref{eq:persistence-stability} and \eqref{eq:empirical-bottleneck} where their hypotheses hold.

\paragraph{Discrete topological object used below.}
On each rectangular grid a selected cell centre is a vertex and edges join vertices at Chebyshev distance one: the eight-neighbour ``king'' graph, whose flag complex gives the same $H_0(-;\F_2)$. A four-neighbour graph is evaluated as an adjacency sensitivity test. Requested mass levels retain the complete threshold tie, producing nested superlevel graphs $G_m$ and a realized mass that may slightly exceed the request.

Because the gate is fixed and the credible filtration is nested, the conditioned graphs $G_m$ are nested too, so $\{H_0(G_m;\F_2)\}_m$ with its inclusion-induced maps is a bona fide degree-zero persistence module admitting an exact barcode in the enclosed-mass coordinate. That barcode, computed in Section~\ref{sec:barcode}, is the primitive topological object here. Deleting small components level by level is \emph{not} functorial --- a deleted component can later merge into a retained one, so independent thresholding need not commute with the inclusion maps --- and the filtered count is therefore retained only as a derived robustness summary,
\begin{equation}
 \widehat\beta_0^{\mathrm{post}}(m)
 =\beta_0\!\left(\mathcal F_{3,10^{-4}}[G_m];\F_2\right),
 \label{eq:postprocessed-beta0}
\end{equation}
and every topological claim is stated in terms of the barcode instead.

\begin{remark}[Four distinct counts]
\label{rem:beta-zero}
Four counts are kept apart: $\beta_0(G_m)$, the components of the declared finite graph, which the barcode decomposes; $\widehat\beta_0^{\mathrm{post}}(m)$ of Eq.~\eqref{eq:postprocessed-beta0}, a post-processed sensitivity summary; $\beta_0(\Pi^m)$, the Betti number of the continuum set, which is never computed and never claimed to equal the others; and $|\Lambda|$, the number of declared physical labels, fixed by the descriptor and independent of any mesh. The content of the benchmarks is the relation between the first and the last under refinement.
\end{remark}

The PCR3BP grid uses the same eight-neighbour graph only for the one-cell boundary band; its class probabilities come from first-hit labels and quadrature, not from graph components.

Implementation scaling is summarized in Appendix~\ref{app:complexity}.

\section{Numerical orbital benchmark}
\label{sec:numerical}

Orbital inference combines incomplete angular observations, hard admissibility
and long-horizon event preimages. This synthetic short-arc example couples an
admissible region, a simplified manifold of variations (MOV), a labelled
transport-conditioned filtration and a noisy follow-up experiment. We set the inference epoch to $t=0$.

\subsection{Dynamical and observational model}

We use planar Earth two-body dynamics with gravitational parameter $\mu_E=398600.4418\ \mathrm{km^3\,s^{-2}}$, Earth radius $R_E=6378.137\ \mathrm{km}$ and an equatorial station rotating at $\omega_E=7.2921159\times10^{-5}\ \mathrm{rad\,s^{-1}}$. At the reference epoch, the topocentric state is parameterized by an angular attributable $(\alpha,\dot\alpha)$ and systematic-ranging coordinates $(\rho,\dot\rho)$:
\begin{align}
 r_0&=q_0+\rho u(\alpha),\\
 v_0&=\dot q_0+\dot\rho u(\alpha)+\rho\dot\alpha u_\perp(\alpha).
 \label{eq:numerical-state}
\end{align}
The truth is $(\alpha,\dot\alpha,\rho,\dot\rho)=(50^\circ,1.5\times10^{-4}\ \mathrm{rad\,s^{-1}},20000\ \mathrm{km},0)$, corresponding to an elliptic orbit with semimajor axis $18617.9\ \mathrm{km}$, eccentricity $0.3325$ and period $7.02\ \mathrm{h}$. Seven topocentric angles are generated over $120\ \mathrm{s}$ with independent Gaussian noise of $2$ arcsec and fixed random seed 4.

The reference systematic-ranging grid contains $101\times101$ points over
\begin{equation}
 12000\le\rho\le35000\ \mathrm{km},
 \qquad
 -2.5\le\dot\rho\le2.5\ \mathrm{km\,s^{-1}}.
 \label{eq:numerical-grid}
\end{equation}
The admissible region requires a bound orbit, pericentre altitude at least $200\ \mathrm{km}$ and apocentre radius at most $50000\ \mathrm{km}$. At each admissible grid point, $(\alpha,\dot\alpha)$ is corrected by nonlinear least squares. The resulting map $q=(\rho/10^4\ \mathrm{km},\dot\rho/(1\ \mathrm{km\,s^{-1}}))\mapsto(\alpha,\dot\alpha,\rho,\dot\rho)$ parameterizes the planar corrected MOV. The baseline profile measure relative to uniform $q$-area is
\begin{equation}
 w_{ij}=Z^{-1}\exp\left[-\frac12\left(\chi^2_{ij}-\chi^2_{\min}\right)\right].
 \label{eq:numerical-weights}
\end{equation}
We also compute two reference-measure corrections requested by the framework. For induced MOV volume, embed the surface in the declared dimensionless coordinates
\[
 s(q)=\left(\alpha/(1\ \mathrm{rad}),\dot\alpha/(10^{-4}\ \mathrm{rad\,s^{-1}}),q_1,q_2\right)
 \]
\begin{equation} 
 J_{\mathrm{MOV}}(q)=\sqrt{\det(Ds(q)^{\mathsf T}Ds(q))},
 \label{eq:MOV-volume}
\end{equation}
and multiply Eq.~\eqref{eq:numerical-weights} by $J_{\mathrm{MOV}}$. For an approximate marginal posterior relative to uniform $(\rho,\dot\rho)$ area and a locally flat attributable prior, Laplace integration of $(\alpha,\dot\alpha)$ multiplies the profile likelihood by $\det(J_a^{\mathsf T}J_a)^{-1/2}$, where $J_a$ is the standardized residual Jacobian with respect to the attributable. These three measures are each normalized before mass ranking. At $\grid{101}$ nodes, $\chi^2_{\min}=5.356$ for five residual degrees of freedom and 2782 nodes are admissible. In three-dimensional orbit determination the admissible region and MOV have higher dimension, but the conclusions expected to survive are the need for an explicit induced/reference measure, mass-level tracking and label-boundary refinement; the number and physical meaning of return classes are model-specific and are not extrapolated from this planar test.

\begin{table*}[t]
\caption{\footnotesize Declared settings of the short-arc benchmark.}
\label{tab:benchmark-settings}
\centering
\small
\begin{tabularx}{0.94\textwidth}{@{}p{0.34\textwidth}X@{}}
\toprule
Quantity & Value \\
\midrule
Observational arc & 7 angular measurements in 120 s \\
Angular noise & 2 arcsec, independent Gaussian \\
Systematic-ranging domain & $\rho\in[12000,35000]$ km, $\dot\rho\in[-2.5,2.5]$ km s$^{-1}$ \\
Admissibility & elliptic; $r_p\ge R_E+200$ km; $r_a\le50000$ km \\
Grid and correction & $101\times101$ points; nonlinear correction of $(\alpha,\dot\alpha)$ \\
Propagation & fourth-order Runge--Kutta; 30 s step to the 12 h gate \\
Inferential status & profile, induced-MOV-volume and Laplace-marginal measures, all with declared dimensionless/reference conventions \\
Target epoch & $T=12$ h \\
Reacquisition gate & $|\operatorname{wrap}(\alpha_T-\alpha_T^\star)|\le15^\circ$, $|\rho_T-\rho_T^\star|\le12000$ km \\
Connectivity & eight-neighbour graph; post-processing retains components with at least three cells and profile mass $10^{-4}$; four-neighbour sensitivity also reported \\
\bottomrule
\end{tabularx}
\end{table*}

\subsection{Admissible region, MOV and reference-measure conventions}

Figure~\ref{fig:shortarc-geometry}(a) shows the admissible region obtained from the measured attributable. Panel (b) shows the profile likelihood on the corrected MOV, three profile-mass boundaries, and the 95\% ellipse of the moment-matched covariance in $(\rho,\dot\rho)$. The covariance ellipse is a valid second-moment summary, but it is a single connected set and spans states with different future return behaviour.

Let $F_T$ map an initial corrected orbit to topocentric angle and range at $T=12$ h. The reacquisition gate $H$ is centered on the truth output and is given in Table~\ref{tab:benchmark-settings}. The conditioned filtration is
\begin{equation}
 \Pi_{0,T}^m(H)=C_0^m\cap F_T^{-1}(H).
 \label{eq:numerical-pathway}
\end{equation}
At $\grid{101}$ nodes and requested 99.5\% profile mass, the complete-tie mask realizes mass $0.995008$ and the gate carries $0.097670$. Three mass levels must be distinguished and are reported separately throughout: the \emph{requested} mass, the \emph{realized} tie-complete mass, and the continuum \emph{regular} mass of Eqs.~\eqref{eq:enclosed-mass-map}--\eqref{eq:credible}. Because the law of the index is atomic the third does not exist here; the realized level is a finite rank in the sense of Corollary~\ref{cor:atomic-mass-rank}. The two components carry $0.052629$ and $0.045041$ and have profile-weighted mean horizon-to-period ratios $1.73$ and $0.85$; Figure~\ref{fig:shortarc-geometry}(c) shows them. The covariance ellipse contains neither their labels nor the non-event gap.

\begin{figure*}[pos=t]
\centering
\includegraphics[width=0.96\textwidth]{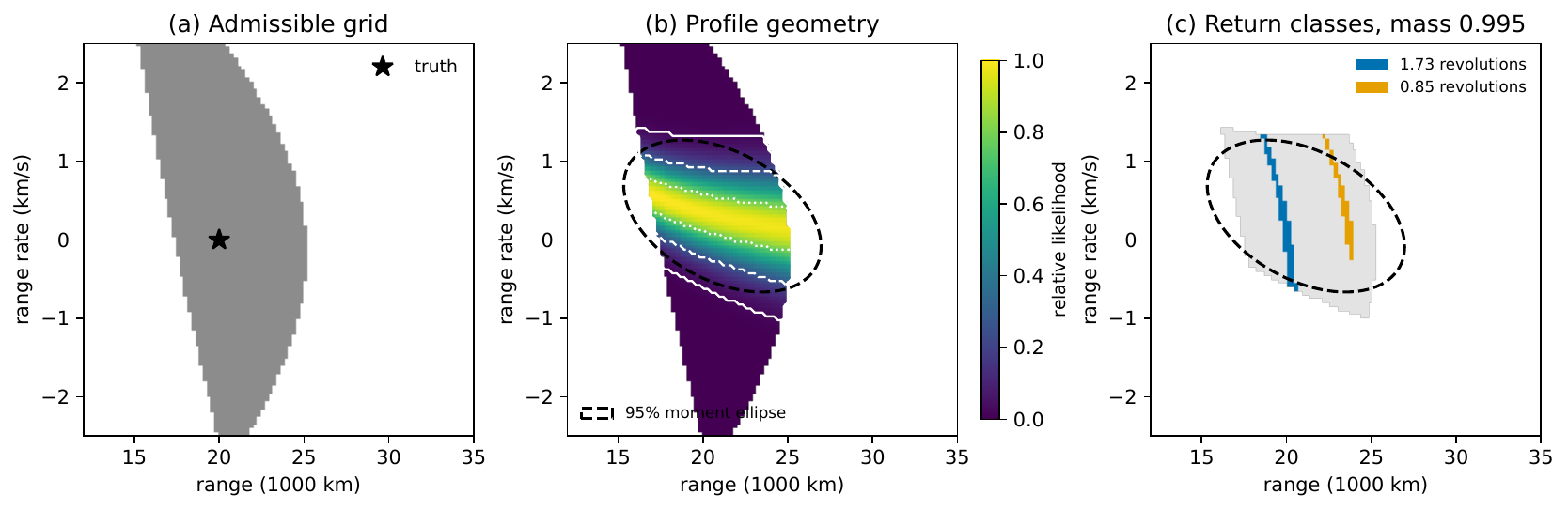}
\caption{\footnotesize Short arc geometry. (a) Admissibility. (b) Corrected MOV, profile-mass boundaries and moment-matched covariance ellipse. (c) The two gate-preimage classes at requested mass 0.995, labelled by their profile-weighted mean number of revolutions.}
\label{fig:shortarc-geometry}
\end{figure*}

Table~\ref{tab:shortarc-measures} compares the three reference-measure conventions at the declared scaling. Volume correction and Laplace marginalization preserve two components and alter either class mass by less than $1.5\times10^{-5}$, so the reported topology is not an artifact of profile weighting in this benchmark. Two qualifications are needed, and neither was stated in the previous version. First, the agreement is a computed result for this configuration, not a general equivalence of profiling and marginalization. Second, and more importantly, the induced volume of Eq.~\eqref{eq:MOV-volume} is an area element of an ambient metric, and that metric is fixed only once the angular coordinates are made dimensionless: ``induced MOV volume'' is therefore itself a declared convention. At the declared scaling the correction is nearly inactive, with 5--95\% quantiles $1.000031$--$1.000145$, so Table~\ref{tab:shortarc-measures} is a low-power test.

The lower half of Table~\ref{tab:shortarc-measures} makes the convention explicit by varying the angular-rate scale over four decades. The volume factor ranges from $1+2\times10^{-8}$ at $10^{-2}\ \mathrm{rad\,s^{-1}}$ to a median of $1.448$ with maximum $2.269$ at $10^{-6}\ \mathrm{rad\,s^{-1}}$. Even that factor-two reweighting leaves the component count at two, while moving the smaller class mass by $9.2\%$. The count is thus robust to the convention over the whole tested range, whereas the masses are not, and the declared scaling must accompany any quoted mass.

\begin{table*}[t]
\centering
\small
\caption{\footnotesize Reference-measure conventions at requested mass 0.995 on the $\grid{101}$ grid. Top: the three measures at the declared angular scaling of Eq.~\eqref{eq:MOV-volume}. Bottom: dependence of the induced-volume convention on the declared angular-rate scale, whose third entry is that same declared scaling.}
\label{tab:shortarc-measures}
\begin{tabular}{@{}lrrrr@{}}
\toprule
Measure & Gate mass & Components & Class 1 mass & Class 2 mass \\
\midrule
Profile counting & 0.097670 & 2 & 0.052629 & 0.045041 \\
Induced MOV volume & 0.097670 & 2 & 0.052629 & 0.045041 \\
Laplace marginal & 0.097661 & 2 & 0.052635 & 0.045027 \\
\midrule
Rate scale (rad s$^{-1}$) & $J$ 5\%/50\%/95\% & Components & \multicolumn{2}{c}{Class 2 mass} \\
\midrule
$10^{-2}$ & 1.000000/1.000000/1.000000 & 2 & \multicolumn{2}{c}{0.045041} \\
$10^{-3}$ & 1.000000/1.000001/1.000001 & 2 & \multicolumn{2}{c}{0.045041} \\
$10^{-4}$ & 1.000031/1.000055/1.000145 & 2 & \multicolumn{2}{c}{0.045041} \\
$10^{-5}$ & 1.003069/1.005473/1.014442 & 2 & \multicolumn{2}{c}{0.044958} \\
$10^{-6}$ & 1.270746/1.448331/1.977202 & 2 & \multicolumn{2}{c}{0.040905} \\
\bottomrule
\end{tabular}
\end{table*}

\subsection{The transport-conditioned barcode}
\label{sec:barcode}

The gate is a fixed set, so the conditioned sets of Eq.~\eqref{eq:numerical-pathway} are nested in $m$ and their degree-zero homology is a genuine persistence module. Its barcode is computed exactly, with no component filtering, by a union--find sweep in the mass coordinate: each admissible cell $x$ carries its complete-tie rank $\rho(x)=\mu\{w\ge w(x)\}$, cells are inserted in increasing $\rho$ with tied cells inserted simultaneously, and the elder rule assigns each merge to the younger component. Bars are reported in enclosed-mass birth/death coordinates.

The outcome is unusually clean and is given in the middle columns of Table~\ref{tab:shortarc-refinement}. From $\grid{75}$ onward the module has exactly two \emph{essential} bars, that is, two components that are born and never merge, and every finite bar has length exactly zero. A zero-length bar is the trivial event of a newly inserted cell joining a component that already exists at that mass level; the module therefore contains no nontrivial finite bar at all, and
\[
 \beta_0\!\left(\Pi_{0,T}^m(H)\right)=2
 \qquad\text{for every }m\in[0.0837,\,0.9950]
\]
on the $\grid{101}$ grid, with the two births at $m=0.0291$ and $m=0.0837$. The two essential classes carry mean revolution counts $1.734$ and $0.849$, so they are exactly the two physical return classes. No threshold, no filter and no smallest-feature convention enters this statement.

The coarse $\grid{51}$ mesh behaves differently and instructively: it has eight essential bars rather than two. Their revolution counts are $1.71, 0.85, 1.76, 0.87, 1.78, 0.88, 1.82, 0.89$, so the six extra components are fragments of the same two physical classes, carrying $0.012$, $0.0088$, $0.0021$, $0.0014$, $0.0002$ and $0.00015$ posterior mass. This is worth stating plainly, because it delimits what persistence can do here: bar length does not discriminate the fragments, since all eight bars are essential. What discriminates them is refinement and class mass, not persistence. The barcode is nevertheless the right primitive, because it makes the failure visible and functorial rather than hiding it inside a filter.

By Corollary~\ref{cor:atomic-mass-rank}, the bottleneck distance between the $\grid{101}$ and $\grid{201}$ barcodes is at most the $b_h=0.07395$ of Table~\ref{tab:mass-rank-check}. The observed displacement of the two essential births is $0.0053$ and $0.0079$, well inside that budget.

\begin{figure*}[pos=t]
\centering
\includegraphics[width=0.96\textwidth]{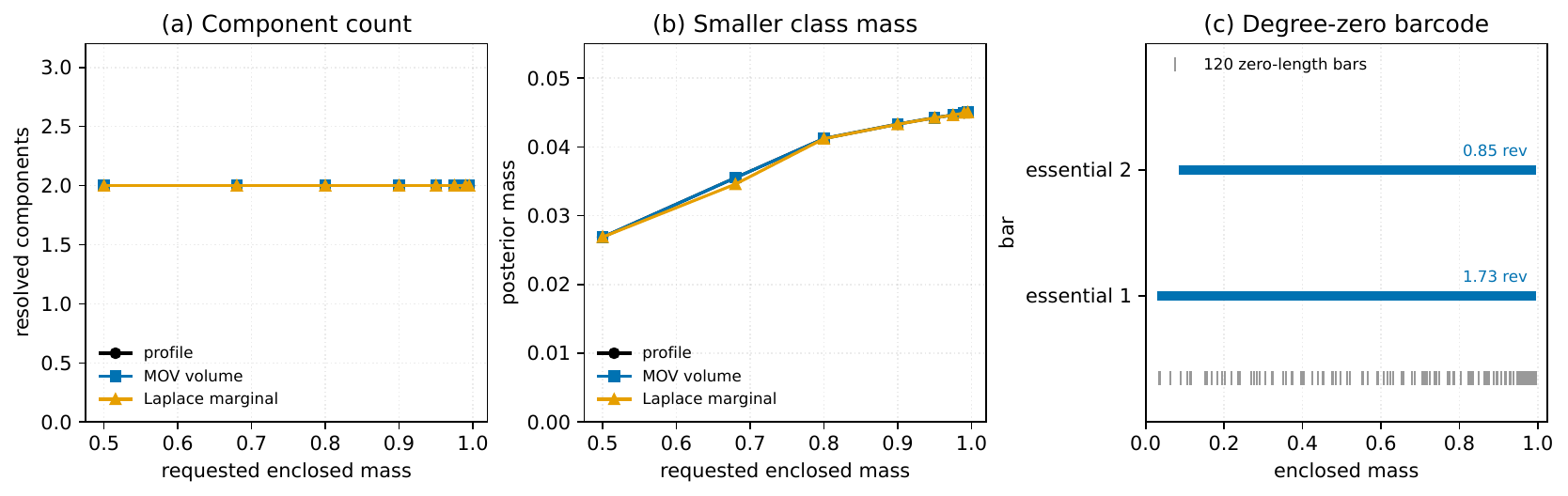}
\caption{\footnotesize Mass filtration and barcode on the $\grid{101}$ grid. (a) Component count and (b) smaller class mass across the tested mass range for the three reference-measure conventions. (c) The degree-zero barcode: two essential bars, annotated by their mean revolution count, and 120 zero-length bars.}
\label{fig:shortarc-barcode}
\end{figure*}

\subsection{Grid refinement and comparison with a label histogram}

Table~\ref{tab:shortarc-refinement} contrasts the barcode and the raw graph counts with the elementary baseline, a two-bin histogram of the revolution label. The histogram returns two classes on every mesh, including $\grid{51}$ where the gate set is in eight pieces. This is not because the histogram is more accurate: it is right by construction, since it does not attempt to resolve spatial structure at all, and for the same reason it cannot warn the user that the mesh is too coarse to support a spatial statement. The mass-indexed components carry that warning, and refinement removes it from $\grid{75}$ onward under eight-neighbour adjacency and from $\grid{151}$ onward under four-neighbour adjacency. That asymmetry between what the two reports can and cannot signal is the concrete added value here, not a difference in the final class count.

\begin{table*}[t]
\centering
\small
\caption{\footnotesize Barcode, raw graph counts and the elementary baseline at requested mass 0.995. Births are enclosed-mass coordinates of the essential bars; every finite bar has length exactly zero at every tested mesh, so no filter is needed. The histogram cannot distinguish two classes from eight coarse fragments of those same two classes.}
\label{tab:shortarc-refinement}
\begin{tabular}{@{}rrrlrrr@{}}
\toprule
Grid & Histogram & Ess. bars & Essential births & 8-nb $\beta_0$ & 4-nb $\beta_0$ & Two largest masses \\
\midrule
$\grid{51}$  & 2 & 8 & 0.0611, 0.0914, $\ldots$, 0.9946 & 8 & 8 & 0.036, 0.029 \\
$\grid{75}$  & 2 & 2 & 0.0362, 0.0803 & 2 & 7 & 0.053, 0.047 \\
$\grid{101}$ & 2 & 2 & 0.0291, 0.0837 & 2 & 4 & 0.053, 0.045 \\
$\grid{151}$ & 2 & 2 & 0.0245, 0.0756 & 2 & 2 & 0.052, 0.046 \\
$\grid{201}$ & 2 & 2 & 0.0238, 0.0758 & 2 & 2 & 0.052, 0.047 \\
\bottomrule
\end{tabular}
\end{table*}

\subsection{Noise ensemble}

\begin{table*}[t]
\centering
\small
\caption{\footnotesize One-hundred-realization noise and graph robustness at requested mass 0.995 on the $\grid{101}$ grid. All quantities refer to the declared simulation design: one truth, one cadence, one Gaussian noise law, one mesh.}
\label{tab:shortarc-noise}
\begin{tabular}{@{}lr@{}}
\toprule
Quantity & Result \\
\midrule
Realizations with two eight-neighbour components & 100/100 \\
Clopper--Pearson 95\% interval for that proportion & $[0.9638,\,1]$ \\
Four-neighbour component counts & $2$: 84, $3$: 10, $4$: 5, $7$: 1 \\
Smaller class mass, 5/25/50/75/95\% & 0.00295/0.01415/0.02780/0.04103/0.04846 \\
Fraction with smaller mass $\ge0.01$ / $\ge0.02$ & 0.80 / 0.66 \\
Gate mass 5/50/95\% & 0.08092/0.08842/0.09979 \\
\bottomrule
\end{tabular}
\end{table*}

\begin{figure*}[pos=t]
\centering
\includegraphics[width=0.96\textwidth]{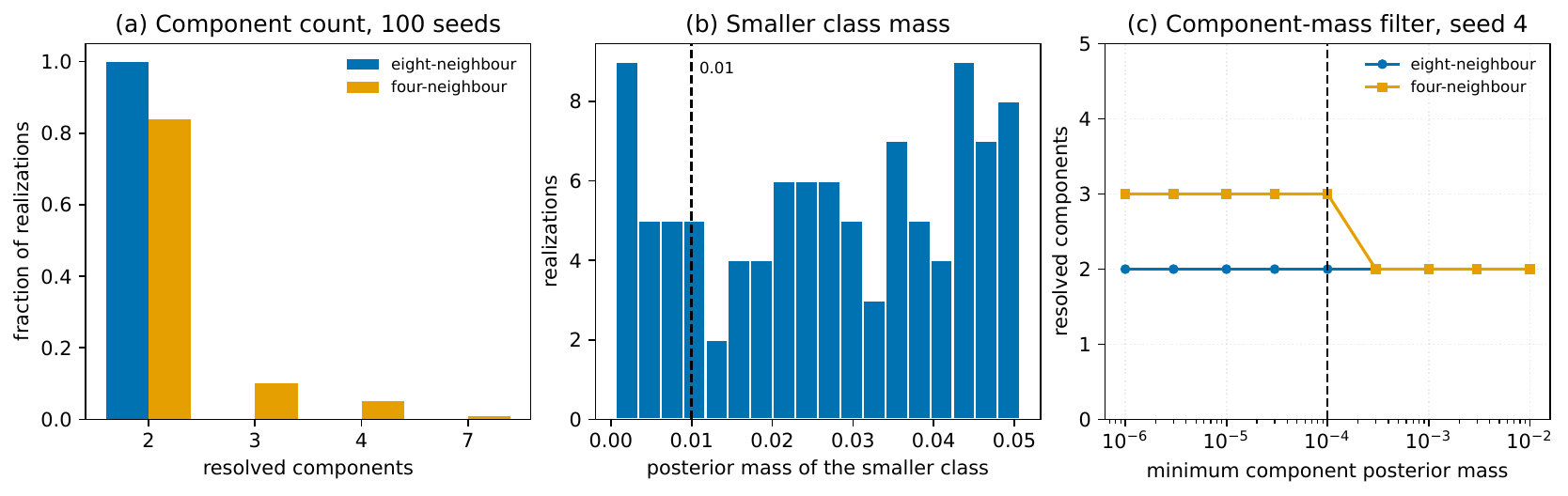}
\caption{\footnotesize Noise and post-processing diagnostics. (a) Component counts over 100 independent 2-arcsec noise realizations, under both adjacencies. (b) Empirical distribution of the smaller class mass. (c) At seed 4, after the three-cell rule has been applied, the eight-neighbour count is two throughout the tested mass-filter range whereas the four-neighbour count falls from three to two when the minimum component mass rises from $10^{-4}$ to $3\times10^{-4}$. Before the three-cell rule the raw four-neighbour count at this mesh is four, as reported in Table~\ref{tab:shortarc-refinement}.}
\label{fig:shortarc-noise}
\end{figure*}

Every one of the 100 tested realizations gives two eight-neighbour components. This is evidence under one truth, one observation cadence, one Gaussian noise law, one mesh and one adjacency, and it is reported as such: zero failures in 100 independent trials is consistent with any per-realization success probability above $0.9638$ at the 95\% Clopper--Pearson level, and no population-level guarantee is intended. The practical weight of the second class is not constant --- 20\% of realizations put less than $0.01$ posterior mass there and 34\% less than $0.02$ --- so robustness of multiplicity is not robustness of operational relevance. The four-neighbour distribution shows that graph choice remains material at the reference mesh. The ensemble uses exact Kepler propagation for speed; against the production Runge--Kutta calculation at seed 4 the maximum angle and range differences are $6.0\times10^{-6}$ rad and $0.008$ km, far below the gate widths. Gate factors $0.8$--$1.2$ retain two eight-neighbour classes.

\subsection{Instantiating the future-experiment block}
\label{sec:design}

For a follow-up at elapsed time $\tau$, define the wrapped-angle output
\begin{equation}
 Z_\tau^{(\sigma)}=\operatorname{wrap}[\alpha_\tau(x)+\xi],
 \,\, \xi\sim\mathcal N(0,\sigma^2),
 \,\, \sigma\in\{2\ \mathrm{arcsec},40^\circ\}
 \label{eq:shortarc-future-experiment}
\end{equation}
The 2-arcsec case represents the original precise sensor; $40^\circ$ is a deliberately degraded planning channel chosen to avoid immediate label-information saturation. The observable is genuinely circular, so the likelihood is the wrapped normal $p(z\mid x)=\sum_k\mathcal N(z;\alpha_\tau(x)+2\pi k,\sigma^2)$ on $[-\pi,\pi)$, with the alias sum truncated at $|k|\le8$. The label variable $\Lambda$ is defined only on the two gate components, so every mutual information below is conditional on the reacquisition-gate event, whose mass is $0.097670$; the conditional label entropy is $0.995642$ bit.

Three conventions are audited rather than asserted. The information integral uses the periodic rectangle rule, spectrally accurate here and stable to ten digits from 128 nodes upward, with 2048 nodes used. An independent Gauss--Hermite mixture-expectation scheme, whose resolution is set by the kernel width rather than by the circle, agrees to $1.6\times10^{-6}$ bit at every tested epoch and is used for the precise sensor, which the circle rule cannot resolve. Alias terms are included and matter at the reported precision: omitting them shifts the degraded value by $7.2\times10^{-5}$ bit at $\tau=8.75$ h.

The Fisher pullback is evaluated at the profile mode by re-propagating perturbed states in the declared dimensionless coordinates $q=(\rho/10^4\ \mathrm{km},\dot\rho/(1\ \mathrm{km\,s^{-1}}))$, not by differencing grid neighbours. Because the observable is circular, each one-sided increment is taken against the centre value, and a step is rejected as wrap-unsafe when either increment exceeds $0.9\pi$; a sweep over $h\in[2.5\times10^{-4},5\times10^{-2}]$ accompanies every value. The audit matters: at $\tau=12$ h the map $\alpha_\tau(q)$ has very large curvature, $h=5\times10^{-2}$ is wrap-unsafe, and the value drifts from $1.5\times10^{3}$ to $4.96\times10^{4}$ under refinement, converging only at $h=2.5\times10^{-4}$ with $8\times10^{-3}$ relative residual drift; at $\tau=8.75$ h it converges to $105.73$ with residual drift $6\times10^{-6}$. The dominant eigenvalue is meaningful only relative to the declared scaling: unlike mutual information, which is invariant under invertible output reparameterization, it is not an intrinsic scalar under state-coordinate rescaling. That caveat governs the comparison below and is not merely a caption remark.

\begin{figure*}[pos=t]
\centering
\includegraphics[width=0.74\textwidth]{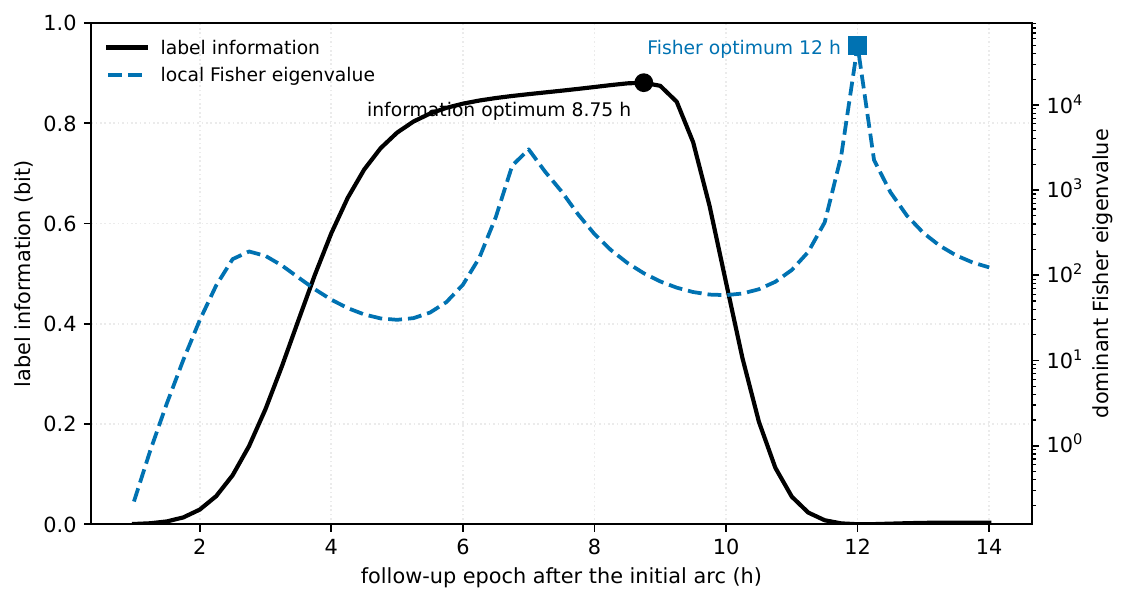}
\caption{\footnotesize Observation-design comparison for the degraded 40-degree channel over 53 epochs at 0.25 h spacing. Label information peaks at $8.75$ h; the dominant local Fisher eigenvalue peaks at 12 h, where the label information has collapsed to $5\times10^{-5}$ bit. Note the logarithmic right-hand axis.}
\label{fig:shortarc-design}
\end{figure*}

The comparison is made over 53 epochs from 1 to 14 h at 0.25 h spacing, so the optima below are optima over a dense grid rather than over six sampled epochs. For the precise sensor the information stays within $0.013$ bit of the label entropy at every tested epoch, with minimum $0.9830$ bit at $13.25$ h, and is therefore saturated: it cannot rank epochs at all. For the degraded channel the information rises to a maximum of $0.8808$ bit at $\tau=8.75$ h and then collapses, reaching $5.1\times10^{-5}$ bit at 12 h. The dominant Fisher eigenvalue behaves in the opposite way: it is $105.7$ at the information optimum and attains its global maximum $4.96\times10^{4}$ at exactly 12 h. The two criteria are not merely offset. Over the 53 epochs their Spearman rank correlation is $-0.008$, that is, indistinguishable from zero.

The mechanism is geometric and explains why the disagreement is structural rather than incidental. The gate of Table~\ref{tab:benchmark-settings} is centred on the truth output at $T=12$ h, so at that epoch both return classes are inside the gate by construction and their class-conditional angle distributions nearly coincide; a coarse sensor then extracts almost nothing about the label. Local sensitivity of $\alpha_\tau$ to the state is simultaneously largest there, because the two classes are being folded together. Moreover, the Fisher pullback of a Gaussian channel scales as $\sigma^{-2}$, so the \emph{epoch} that maximizes it does not depend on sensor quality at all, whereas the epoch that maximizes $I(\Lambda;Z_\tau)$ does. A criterion that cannot see the sensor cannot rank epochs for a task that depends on it. This is Corollary~\ref{cor:fisher-not-sufficient} realized numerically.

\section{PCR3BP benchmark: why transport labels are necessary}
\label{sec:pcr3bp}

The PCR3BP gives a second finite-resolution example. The transit/non-transit classification near the collinear necks goes back to Conley and McGehee, and its global invariant-manifold and tube-dynamics development underlies modern low-energy transport and mission design \cite{Conley1968,McGehee1969,Koon2000,KoonBook2011,JorbaMasdemont1999,GomezMondelo2001,Wiggins1992}. Related Lagrangian constructions connect invariant structures to uncertainty quantification \cite{GarciaSanchez2023,GarciaSanchez2025}. DUG adds posterior-mass incidence and finite-resolution error accounting to these physical transport classes.

\subsection{Model, section and posterior}

In standard synodic nondimensional variables, the equations are
\begin{align}
 \ddot x-2\dot y &= \Omega_x(x,y),
 &
 \ddot y+2\dot x &= \Omega_y(x,y),
 \label{eq:pcr3bp-equations}\\
 \Omega(x,y) &= \frac{x^2+y^2}{2}
 +\frac{1-\mu}{r_1}+\frac{\mu}{r_2},
 & \\
 C_J &=2\Omega-\dot x^2-\dot y^2,
 \label{eq:pcr3bp-jacobi}
\end{align}
where $r_1^2=(x+\mu)^2+y^2$, $r_2^2=(x-1+\mu)^2+y^2$, and $\mu=0.012150585609624$ is the Earth--Moon mass parameter. We set $C_J=3.16$, below the $L_1$ and $L_2$ critical values, so both necks are open. Numerically,
\begin{equation}
 x_{L_1}=0.8369151258,
 \qquad
 x_{L_2}=1.1556821654.
 \label{eq:pcr3bp-lagrange}
\end{equation}

The uncertainty is placed on the fixed-energy Poincar\'e section
\begin{equation}
 \Sigma_C=\left\{(x,0,\dot x,\dot y):
 \dot y>0,
 \;2\Omega(x,0)-\dot x^2-\dot y^2=C_J
 \right\},
 \label{eq:pcr3bp-section}
\end{equation}
parameterized by $(x,\dot x)$. We declare the section-coordinate area $\nu_\Sigma=\dd x\,\dd\dot x$ as reference measure; under reparameterization it must be pushed forward with the probability measure. The density relative to $\nu_\Sigma$ is a Gaussian truncated to the energetically admissible lunar lobe. Its untruncated Gaussian location is
\begin{equation}
 \bar z=(1.11,-0.07),
 \label{eq:pcr3bp-mean}
\end{equation}
with scales $(0.018,0.055)$ and correlation $-0.5$ before truncation and renormalization, on the energetic lobe intersected with $x\in[1.02,x_{L_2}]$, $\dot x\in[-0.26,0.12]$. Conditional-Gaussian quadrature shows the rectangle omits $3.04\times10^{-4}$ of the already truncated mass, and the $\grid{181}$ normalizer differs from the one-dimensional reference integral by $1.38\times10^{-4}$ relatively, reduced to $8.89\times10^{-6}$ on a $\grid{1441}$ audit. Section and epoch are fixed, so we write $C^m$. Its highest-density credible sets are connected at every sampled mass: the finite-grid posterior module contains no class multiplicity.

We propagate to $T=12$ nondimensional units, about $52.2$ days; this captures all but less than $10^{-4}$ of the selected mass on the finest grid, and the lower block of Table~\ref{tab:pcr3bp-refinement} checks $T=10$ and $14$. Labels are assigned solely by the first terminal event: inward crossing of $x=x_{L_1}$ with $\dot x<0$, outward crossing of $x=x_{L_2}$ with $\dot x>0$, or entry into the disk centred at $(1-\mu,0)$ with radius $r_2=0.02$. That disk is a deliberately enlarged close-approach set of about $7.7\times10^3$ km, not the lunar surface, and its radius sensitivity is reported below. ``No terminal event within $T$'' means only that no declared terminal surface was reached before the horizon; it is not a claim of permanent residence. These mutually exclusive rules do not use the stable-manifold trace.

Backward propagation of the stable direction of a planar $L_1$ Lyapunov orbit, obtained by differential correction \cite{Richardson1980}, produces the overlay in Fig.~\ref{fig:pcr3bp-transport}. The overlay is an interpretation aid: it is not an input to classification or quadrature, it is not a uniformly accurate separator of the three finite-time classes, and no quantity reported in this paper depends on it.

\begin{table*}[t]
\vspace*{1.5em}
\centering
\caption{\footnotesize Declared PCR3BP benchmark and grid-quadrature results at 99\% credible mass. State, scale and integration quantities are nondimensional unless days are shown; the close-approach disk is an enlarged artificial set, not the physical lunar surface. Masses are quoted on the finest grid with the label-split intervals of Eq.~\eqref{eq:pcr3bp-split-interval}.}
\label{tab:pcr3bp-settings}
\begin{tabular}{@{}ll@{}}
\toprule
Quantity & Value \\
\midrule
Earth--Moon mass parameter $\mu$ & $0.012150585609624$ \\
Jacobi constant $C_J$ & $3.16$ \\
Poincar\'e coordinates & $(x,\dot x)$ on $y=0$, $\dot y>0$ \\
Untruncated Gaussian location & $(1.11,-0.07)$ \\
Untruncated Gaussian scales & $(0.018,0.055)$ \\
Untruncated Gaussian correlation & $-0.5$ \\
Reference measure & $\nu_\Sigma=\dd x\,\dd\dot x$ \\
Finite horizon & $T=12$ ($52.2$ days); sensitivity at $T=10,14$ \\
Encounter set & center $(1-\mu,0)$, radius $0.02$; artificial close-approach disk \\
$L_1$ first-hit mass; split diagnostic interval & $0.42050$; $[0.38872,0.50872]$ \\
$L_2$ first-hit mass; split diagnostic interval & $0.36665$; $[0.32970,0.44970]$ \\
Close approach to the disk $r_2=0.02$; split interval & $0.20279$; $[0.15033,0.27033]$ \\
No terminal event within $T$ & $5.5\times10^{-5}$ on the $\grid{721}$ grid \\
Omitted fraction of truncated-lobe mass & $3.04\times10^{-4}$ \\
Posterior-weighted 99th percentile of $|\Delta C_J|$ & $7.75\times10^{-7}$ at step $1.5\times10^{-3}$ \\
\bottomrule
\end{tabular}
\end{table*}

\subsection{Connected uncertainty, multiple first-hit outcome classes}

Figure~\ref{fig:pcr3bp-transport}(a) shows two high-density trajectories starting from nearby section states and leaving the lunar region through different necks. Panel (b) displays the first-hit outcome partition on $\Sigma_C$, the 50\%, 90\% and 99\% credible contours, and a computed trace of $W^s(\gamma_{L_1})\cap\Sigma_C$. The posterior geometry alone is a nested family of connected ovals. The transport descriptors cut those ovals into labelled subsets whose destinations differ.

Within the 99\% credible set on the $\grid{181}$ grid with step $1.5\times10^{-3}$, the quadrature masses are $0.422$ for $L_1$, $0.365$ for $L_2$ and $0.203$ for enlarged-disk close approach. Both exit labels occur at every sampled mass from $m=0.50$ to $m=0.995$, while $\beta_0(C^m)=1$. These values locate the finite-grid result; the conservative intervals derived below delimit the quantitative claim.

\begin{figure*}[pos=t]
\centering
\includegraphics[width=\textwidth]{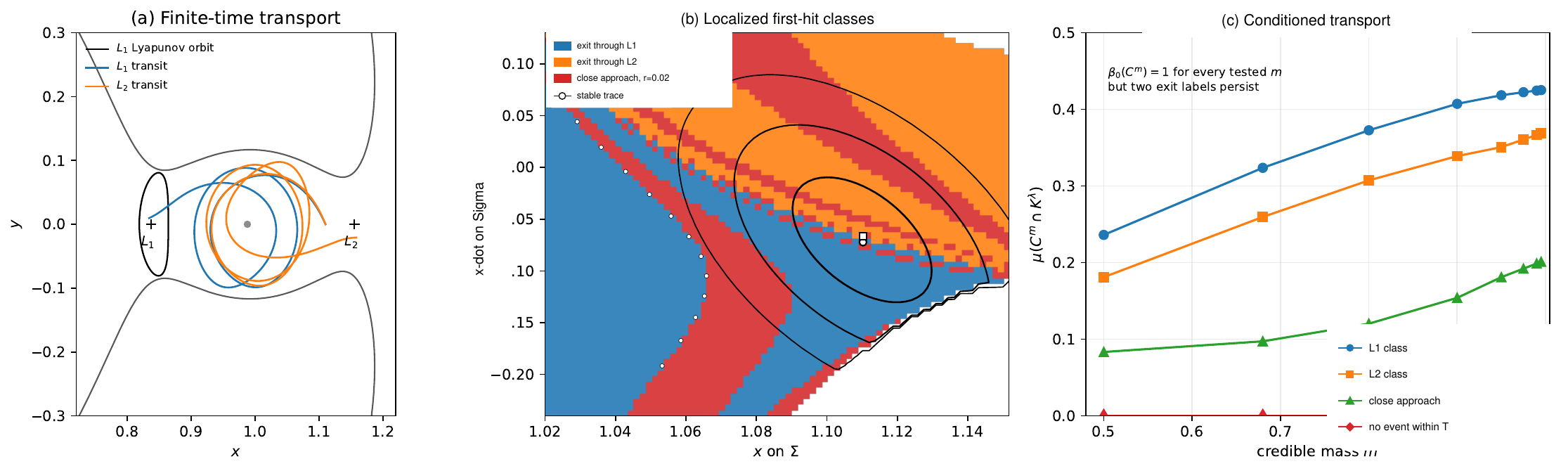}
\caption{\footnotesize PCR3BP transport benchmark. (a) Two nearby states transit through different necks. (b) First-hit classes and credible contours. The stable-manifold trace is an explanatory overlay: it does not generate the labels, and it is not claimed to approximate the whole finite-time class boundary. (c) Mass in the labelled conditioned filtration.}
\label{fig:pcr3bp-transport}
\end{figure*}

The dependency is explicit: the posterior report alone contains one connected high-density family, and the first-hit descriptor supplies the labelled partition. Theorem~\ref{thm:skeleton-obstruction} applies only when the declared future experiment omits those event labels; an event-complete experiment would contain them by construction.

\subsection{Event localization, boundary scaling and class intervals}
\label{sec:pcr3bp-scaling}

The vectorized fourth-order Runge--Kutta calculation stops at the first terminal
set or at $T$. Within every step, all terminal surfaces whose signed event function changes from positive to nonpositive are localized by bisection to an elapsed-time bracket below $7.2\times10^{-10}$. Direction is checked at the localized crossing, and if more than one surface is crossed in a step the earliest localized event is selected. The code records as ambiguous any two localized event times within ten bracket tolerances; the credible mass of such cases is zero in every reported run. Table~\ref{tab:pcr3bp-refinement} changes both the section grid and the fixed step. Step halving leaves the three displayed masses unchanged on the $\grid{181}$ grid and reduces the Jacobi drift quantile by more than an order of magnitude; from $\grid{181}$ to $\grid{271}$ every class mass changes by at most $0.002$.

\begin{table*}[t]
\centering
\small
\caption{\footnotesize PCR3BP refinement and task sensitivity at 99\% credible mass. Upper block: grid and step refinement at $T=12$, $r_2=0.02$. Lower block: horizon and encounter-radius sensitivity on the $\grid{181}$ grid with step $1.5\times10^{-3}$. $B_1$ is the posterior mass in the one-cell eight-neighbour band adjacent to a change of first-hit label; $q_{0.99}$ is the weighted 99th percentile of $|\Delta C_J|$. Ambiguous-event mass is zero in all rows. ``Close app.'' is close approach to the enlarged disk of radius $r_2$.}
\label{tab:pcr3bp-refinement}
\begin{tabular}{@{}lrrrrrr@{}}
\toprule
Setting & $L_1$ & $L_2$ & Close app. & No event & $B_1$ & $q_{0.99}$ \\
\midrule
$\grid{91}$, step $3.0\times10^{-3}$  & 0.420 & 0.365 & 0.205 & 0 & 0.452 & $9.58\times10^{-6}$ \\
$\grid{181}$, step $3.0\times10^{-3}$ & 0.422 & 0.365 & 0.203 & 0 & 0.320 & $1.00\times10^{-5}$ \\
$\grid{181}$, step $1.5\times10^{-3}$ & 0.422 & 0.365 & 0.203 & 0 & 0.320 & $7.73\times10^{-7}$ \\
$\grid{271}$, step $1.5\times10^{-3}$ & 0.420 & 0.367 & 0.203 & $9.9\times10^{-5}$ & 0.248 & $7.75\times10^{-7}$ \\
\midrule
$T=10$, $r_2=0.020$ & 0.4218 & 0.3652 & 0.2026 & 0.00035 & --- & --- \\
$T=14$, $r_2=0.020$ & 0.4219 & 0.3652 & 0.2029 & 0 & --- & --- \\
$T=12$, $r_2=0.018$ & 0.4295 & 0.3820 & 0.1785 & 0 & --- & --- \\
$T=12$, $r_2=0.022$ & 0.4131 & 0.3486 & 0.2283 & 0 & --- & --- \\
\bottomrule
\end{tabular}
\end{table*}

\paragraph{How fast does the unresolved boundary mass decrease?}
The more discriminating diagnostic compares nested grids. A coarse cell is
called \emph{resolved} when its four admissible corners carry one common
first-hit label; $U_h$ is the fine-grid credible mass lying in all other coarse
cells, and $D_h$ is the fine-grid mass whose directly integrated label
disagrees with the common label of a resolved coarse cell. With a single pair
one cannot tell whether $U_h$ is decreasing usefully, so Table~\ref{tab:pcr3bp-disagreement}
now reports four nested comparisons across the five grids
$\grid{46}$, $\grid{91}$, $\grid{181}$, $\grid{361}$ and $\grid{721}$.

\begin{table*}[t]
\centering
\scriptsize
\caption{\footnotesize Nested-grid stability of PCR3BP first-hit labels at 99\% credible mass and step $1.5\times10^{-3}$. $h$ is the coarse section cell width. The last two columns split $U_h$ by the fine-grid label of the unresolved mass.}
\label{tab:pcr3bp-disagreement}
\begin{tabular}{@{}rrrrrrrr@{}}
\toprule
Coarse $\to$ fine & $h$ & $U_h$ & $D_h$ & $L_1$ & $L_2$ & Close app. & $U_h$ split ($L_1$/$L_2$/close) \\
\midrule
$\grid{46}\to\grid{91}$   & $3.02\times10^{-3}$ & 0.44281 & $6.70\times10^{-3}$ & 0.41960 & 0.36538 & 0.20503 & 0.132/0.166/0.145 \\
$\grid{91}\to\grid{181}$  & $1.51\times10^{-3}$ & 0.32053 & $1.46\times10^{-3}$ & 0.42187 & 0.36524 & 0.20289 & 0.096/0.100/0.125 \\
$\grid{181}\to\grid{361}$ & $7.54\times10^{-4}$ & 0.19862 & $1.01\times10^{-4}$ & 0.42083 & 0.36705 & 0.20205 & 0.052/0.063/0.083 \\
$\grid{361}\to\grid{721}$ & $3.77\times10^{-4}$ & 0.11749 & $1.25\times10^{-3}$ & 0.42050 & 0.36665 & 0.20279 & 0.031/0.036/0.051 \\
\bottomrule
\end{tabular}
\end{table*}

\begin{figure*}[pos=t]
\centering
\includegraphics[width=0.86\textwidth]{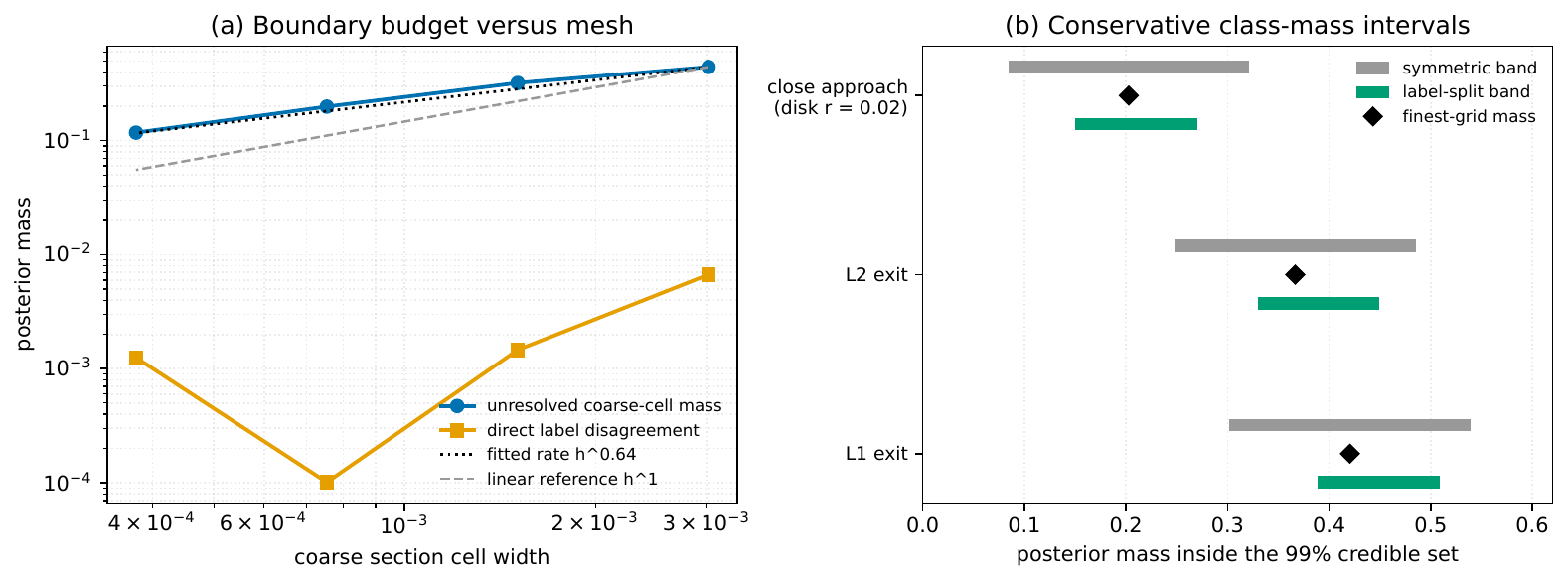}
\caption{\footnotesize (a) Unresolved coarse-cell mass and direct label-disagreement mass against the coarse section cell width, with the fitted rate and a linear reference. (b) Finest-grid class masses with the symmetric band of Eq.~\eqref{eq:pcr3bp-mass-interval} and the label-split band of Eq.~\eqref{eq:pcr3bp-split-interval}.}
\label{fig:pcr3bp-scaling}
\end{figure*}

The unresolved mass does decrease, but slowly and sublinearly: the fitted rate
over the ladder is $U_h\sim h^{0.64}$, and the pairwise exponents
$0.47$, $0.69$, $0.76$ are still drifting upward at the finest pair, so the
exponent should be read as an estimate over this range and not as a limit. The
comparison with the linear reference in Fig.~\ref{fig:pcr3bp-scaling}(a) is the
substantive point. For a boundary that is a rectifiable curve one expects the
one-cell band to carry mass of order $h$; a band mass scaling like $h^{p}$ with
$p<1$ is the signature of a boundary whose covering number grows faster,
corresponding to an effective box dimension $d\simeq2-p\simeq1.36$ over the
tested decade. This is consistent with the interlaced first-hit geometry
visible in Fig.~\ref{fig:pcr3bp-transport}(b) and with classical tube dynamics near
the collinear necks \cite{Conley1968,McGehee1969,Koon2000}, and it explains
directly why a fixed uniform mesh cannot deliver certified continuum class
probabilities here: at rate $h^{0.64}$, reducing the budget by one decade
requires refining each coordinate by a factor of about $36$, that is a
$1.3\times10^{3}$-fold increase in the number of section cells and in the
number of propagations. Adaptive refinement of the label boundary, not a finer
uniform grid, is the appropriate next step.

The direct disagreement $D_h$ is smaller by two to three orders of magnitude
but is not monotone: it falls from $6.7\times10^{-3}$ to $1.0\times10^{-4}$ and
then rises to $1.3\times10^{-3}$ on the finest pair. With four points and a
quantity this small relative to $U_h$ we do not fit a rate to it; it is
reported as a measured bound on observed label reversal, and its
non-monotonicity is itself evidence that individual cell labels near the
boundary are not settling under uniform refinement.

\paragraph{Class-mass intervals.}
For the finest pair define the diagnostic budget $B_h=U_h+D_h=0.11875$ and, for
each terminal class $\lambda$, the symmetric band
\begin{equation}
 I_h^\lambda=
 \bigl[\max\{0,\widehat M_h^\lambda-B_h\},
       \min\{0.99,\widehat M_h^\lambda+B_h\}\bigr].
 \label{eq:pcr3bp-mass-interval}
\end{equation}
Equation~\eqref{eq:pcr3bp-mass-interval} is doubly conservative. It ignores that
the unresolved mass is itself labelled on the fine grid, and it ignores the
constraint $\sum_\lambda M^\lambda=M$, so that the three symmetric bands are not
jointly attainable. Splitting $U_h$ by fine-grid label, as in the last column of
Table~\ref{tab:pcr3bp-disagreement}, gives the sharper band
\begin{equation}
 J_h^\lambda=
 \bigl[\widehat M_h^\lambda-U_h^\lambda-D_h,\;
       \widehat M_h^\lambda+(U_h-U_h^\lambda)+D_h\bigr],
 \label{eq:pcr3bp-split-interval}
\end{equation}
which still assigns every unresolved and every reversed cell against the class
in question, but is automatically compatible with the simplex constraint. Its
width is $U_h+2D_h$ for every class, independent of $\lambda$.

Using the $\grid{721}$ masses, Eq.~\eqref{eq:pcr3bp-split-interval} gives
$[0.3887,0.5087]$ for $L_1$, $[0.3297,0.4497]$ for $L_2$ and
$[0.1503,0.2703]$ for close approach to the enlarged disk of radius $0.02$,
each of width $0.1200$. The corresponding symmetric bands of
Eq.~\eqref{eq:pcr3bp-mass-interval} have width $0.2375$, so on the same finest pair the split
construction halves the reported uncertainty at no cost in rigour: the
symmetric band would admit a close-approach mass as small as $0.084$, the split
band only $0.150$.
These remain \emph{diagnostic} intervals: they are neither confidence intervals
nor certified continuum bounds, and they inherit the enlarged-disk convention.
What is now supported quantitatively is that all three labels occur and that
each carries order-one mass with the stated band; what is not supported is a
three-decimal class probability.

All three labels persist over the tested horizons and radii. The horizon is saturated by $T=12$, whereas changing the artificial radius by $\pm10\%$ moves close-approach mass from $0.179$ to $0.228$ and reallocates exit mass accordingly, so that mass must always be quoted with its declared radius.

\subsection{Direct comparison with baseline reports}
\label{sec:baseline-comparison}

Table~\ref{tab:baseline-comparison} records what changes when the same computations are reduced to standard summaries. The comparison is intentionally task-specific: it does not claim that DUG replaces covariance, Monte Carlo or Fisher information, only that those reduced reports omit one of the quantities needed by the stated channel or design task.

\begin{table*}[t]
\centering
\scriptsize
\caption{\footnotesize DUG versus baseline reports on the two benchmarks. Counts are finite-grid quantities under the declared graph and filter.}
\label{tab:baseline-comparison}
\begin{tabularx}{\textwidth}{@{}p{0.16\textwidth}p{0.24\textwidth}p{0.25\textwidth}X@{}}
\toprule
Case & Baseline output & Quantity omitted by that output & DUG output \\
\midrule
Short arc, $\grid{51}$ & Two-bin revolution-label histogram & Spatial fragmentation and its mesh dependence & Eight essential bars at $\grid{51}$, all carrying one of the same two revolution labels; refinement gives exactly two essential bars from $\grid{75}$ onward \\
Short arc, follow-up design & Dominant local Fisher eigenvalue; optimum $12$ h & Finite label-conditioned output laws & $I(\Lambda;Z_\tau)$ is maximized at $8.75$ h; the two criteria have Spearman rank correlation $-0.008$ over 53 epochs \\
PCR3BP, posterior moments or ellipse & One connected local uncertainty summary & First-hit destination labels & Connected credible sets intersect $L_1$, $L_2$ and enlarged-disk close-approach classes \\
PCR3BP, raw event Monte Carlo & Three first-hit frequencies at one mass and mesh & Connectivity across mass, the $h^{0.64}$ boundary rate and unresolved-mass accounting & Labelled mass filtration plus $U_h$, $D_h$ and the split intervals of Eq.~\eqref{eq:pcr3bp-split-interval} \\
\bottomrule
\end{tabularx}
\end{table*}

\section{Discussion and conclusions}
\label{sec:discussion}

Theorem~\ref{thm:persistence-stability} is the mathematical centre of the paper: it is the one statement that controls a filtration whose index is itself a functional of the estimated measure. The continuous tilt example applies all of its hypotheses directly and shows that the concentration modulus, not total variation, dominates the conservative bound; Corollary~\ref{cor:atomic-mass-rank} makes the short-arc quantity $b_h$ a rigorous comparison of two weighted-grid filtrations, and Section~\ref{sec:theorem-scope} states exactly where that control stops.

\subsection{Scope and limitations}
\label{sec:scope}

The scope is a mass-indexed reporting framework with one continuum stability theorem, its refinement and atomic corollaries, and two reproducible finite-resolution demonstrations. The short-arc topological claims are statements about the declared finite graph filtration; all 100 tested noise realizations give two components under eight-neighbour adjacency, but four-neighbour fragmentation and the smaller class mass remain sensitive, and the three reference-measure conventions agree here, not in general. The two-dimensional parameterization is computational: a full orbit-determination pipeline would require the same measure and refinement audits on the corrected manifold, observation biases and higher-dimensional nuisance parameters.

In the PCR3BP, $0.117$ posterior mass still lies in coarse-unresolved cells at the finest tested pair and the unresolved mass decreases only like $h^{0.64}$, so Eq.~\eqref{eq:pcr3bp-split-interval} supports finite-grid label multiplicity with order-one weights, not continuum class probabilities; the fitted exponent itself rests on four points and is still drifting. The close-approach result is conditional on an enlarged disk, and ``no terminal event within $T$'' is only a finite-horizon label. The planar calculation says nothing about the spatial problem, which would need a four-dimensional energy-section representation, posterior-informed coordinates and an independent event-boundary audit.

\subsection{Conclusions}
\label{sec:conclusion}

DUG supplies a common enclosed-mass domain for the persistence of posterior credible sets and transport-conditioned subsets, together with a declared future experiment. Its analytical result is the mass-rank interleaving bound with its refinement and atomic corollaries; the coordinate, isotopy and transport statements are supporting checks.

The reproducible calculations support three finite-resolution statements. First, the short-arc transport-conditioned module has exactly two essential bars and no nontrivial finite bar from $\grid{75}$ onward, so ``two return classes'' is a statement about an unfiltered persistence module rather than about a post-processed count, and it survives 100 noise realizations under eight-neighbour adjacency, with the smaller class mass varying substantially and four-neighbour adjacency remaining fragile. Second, for a degraded follow-up sensor the label-information and local-Fisher criteria have essentially zero rank correlation over 53 epochs, and the Fisher optimum falls exactly where the label information vanishes. Third, in the PCR3BP the three first-hit labels persist under grid, step, horizon and radius tests, while the unresolved boundary mass falls only like $h^{0.64}$, so uniform refinement cannot deliver certified continuum class probabilities and the split intervals of Eq.~\eqref{eq:pcr3bp-split-interval} are what the data support. Persistence of credible sets, future-law information and first-hit organization answer different questions and should be coupled only through a declared scientific task.

\section*{Declaration of competing interest}
The author declares no competing financial or personal interests that could have influenced the work reported in this paper.

\section*{Funding}
This research received no specific grant from any funding agency in the public, commercial, or not-for-profit sectors.

\section*{Data availability}
The examples are synthetic and use no external observational data. Supplementary Software S1, supplied with this submission, is the complete reproducibility archive: it contains the manuscript source, figure-generation code, exact dependency versions and CSV outputs for every numerical table and figure, including the continuous mass-rank and nested-grid PCR3BP diagnostics. Short-arc noise seeds are declared in the scripts and tables; PCR3BP calculations are deterministic. The archive contains a \texttt{README} listing every script, the table or figure it produces and the exact dependency versions, together with a single \texttt{run\_all.sh} that regenerates every CSV and figure in dependency order; a \texttt{--quick} flag on the expensive scripts gives a fast installation check. The submission-system record is the review-stage identifier for S1; the unchanged archive will receive a public archival DOI upon acceptance.

\section*{Declaration of generative AI and AI-assisted technologies in the writing process}
During the preparation of this work the author used OpenAI ChatGPT in order to assist with language editing and mathematical exposition. After using this tool/service, the author reviewed and edited the content as needed and takes full responsibility for the content of the publication.

\appendix

\section{Supporting coordinate, transport and local-limit results}
\label{app:supporting-results}

This appendix records the structural checks used by the operational interface; they are standard consequences of pushforward covariance, the chain rule and Laplace asymptotics. Let $\psi:\X\to\X'$ be a diffeomorphism and set $\mu'=\psi_\#\mu$, $\nu'=\psi_\#\nu$. The Radon--Nikodym chain rule gives
\begin{equation}
 \frac{\dd\mu'}{\dd\nu'}(x')=
 \frac{\dd\mu}{\dd\nu}(\psi^{-1}x')
 \quad \nu'\text{-a.e.}
 \label{eq:appendix-density-covariance}
\end{equation}
Consequently posterior ranks are unchanged, $\rho'(\psi x)=\rho(x)$, and $C^{\prime m}=\psi(C^m)$. The induced maps on homology commute with the mass-inclusion maps, so the credible and label-conditioned persistence modules are isomorphic when labels and reference measures are transported with the state.

If $P_x$ is differentiable in quadratic mean and the state coordinate changes by $x'=\psi(x)$, score directions transform by $u=D\psi_x^{-1}u'$. Hence the Fisher matrix obeys
\begin{equation}
 G'(x')=D\psi_x^{-\mathsf T}G(x)D\psi_x^{-1},
 \label{eq:appendix-fisher-covariance}
\end{equation}
the coordinate law of a covariant two-tensor. An invertible measurable change of output coordinates leaves both the Fisher tensor and mutual information unchanged; a non-invertible reduction can only decrease the latter by data processing.

For a deterministic diffeomorphic flow, material transport is exact: if $\mu_s=(\Phi_{t,s})_\#\mu_t$ and $\nu_s=(\Phi_{t,s})_\#\nu_t$, then $C_s^m=\Phi_{t,s}(C_t^m)$. Homology is therefore unchanged by forecast alone. A Bayesian assimilation multiplies the forecast density by a likelihood and renormalizes it, so its credible filtration must instead be recomputed.

Finally, suppose on a compact admissible set
\begin{equation}
 r_\varepsilon(x)=Z_\varepsilon^{-1}a(x)
 \exp[-V(x)/\varepsilon],
 \label{eq:appendix-laplace-posterior}
\end{equation}
where $a$ is positive and $C^2$, $V$ is $C^3$ with a unique global minimum $x_\star$, and $H=\nabla^2V(x_\star)\succ0$. For every fixed regular mass $m\in(0,1)$, Laplace's method gives
\begin{align}
 \varepsilon^{-1/2}(C_\varepsilon^m-x_\star)
 &\longrightarrow
 \{z:z^{\mathsf T}Hz\le\chi^2_{n,m}\},
 \label{eq:appendix-gaussian-set}\\
 \operatorname{Cov}_{\mu_\varepsilon}(x)
 &=\varepsilon H^{-1}+o(\varepsilon).
 \label{eq:appendix-gaussian-covariance}
\end{align}
The first convergence is local Hausdorff convergence after rescaling. This recovers covariance ellipsoids as a local fixed-mass limit, not as a determination of the global filtration or its transport labels. Proof details are given in Appendices~\ref{app:gaussian} and \ref{app:metric-coordinate}.

\section{Task-factorization witnesses}
\label{app:task-witnesses}

For Proposition~\ref{prop:mi-minimality}, take four states with positive weights, binary labels and binary outputs. The joint law is
\begin{equation}
 p_{\lambda z}=\sum_{i:\lambda(i)=\lambda}w_iP_i(z),
 \label{eq:appendix-mi-factorization}
\end{equation}
so weights, labels and conditional laws suffice. Each is essential. With deterministic $Z=\Lambda$, varying the total label weight changes the binary entropy. With balanced fixed weights and labels, state-independent Bernoulli$(1/2)$ laws give zero mutual information whereas $P_i=\delta_{\lambda(i)}$ gives one bit. With uniform weights and conditional laws $(\delta_0,\delta_0,\delta_1,\delta_1)$, labels $(0,0,1,1)$ give one bit while crossed labels $(0,1,0,1)$ give zero. The first two future-law families can be locally constant on four separated neighborhoods, so both have zero local Fisher tensor while their finite label information differs. This proves Corollary~\ref{cor:fisher-not-sufficient}.

For Proposition~\ref{thm:skeleton-obstruction}, let the common credible set be $D_2=\{(x,y):x^2+y^2\le4\}$ and take $T=1$. The setting of Section~\ref{sec:setting} requires genuine two-parameter flows satisfying the cocycle property, not merely a pair of terminal diffeomorphisms, so the witnesses are given as vector fields. On $\R^2$ define the smooth, complete, nonautonomous fields
\begin{equation}
 W^A(s,x,y)=(1,0),
 \,\,
 W^B(s,x,y)=\bigl(1,\;3\left[1-(x-s)^2\right]\bigr)
 \label{eq:appendix-shear-fields}
\end{equation}
Along any solution of $W^B$ starting from $(x,y)$ at time $s_0$ one has $x(s)=x+(s-s_0)$, so the second component of $W^B$ evaluates to the constant $3[1-(x-s_0)^2]$ and the flows integrate in closed form:
\[
 \Phi^A_{s_0,s}(x,y)=(x+(s-s_0),\,y),
 \]
 \begin{equation}
 \Phi^B_{s_0,s}(x,y)=\bigl(x+(s-s_0),\;y+3\left[1-(x-s_0)^2\right](s-s_0)\bigr).
 \label{eq:appendix-shear-flows}
\end{equation}
Both families are $C^\infty$ in all arguments, are diffeomorphisms of $\R^2$ for every pair $(s_0,s)$, and satisfy $\Phi_{s_1,s_2}\circ\Phi_{s_0,s_1}=\Phi_{s_0,s_2}$, as required by Eq.~\eqref{eq:flow}: for $W^B$ the increment accumulated on $[s_0,s_1]$ and on $[s_1,s_2]$ carries the same constant $3[1-(x-s_0)^2]$, because the field is designed to be evaluated along the moving point. Their time-one maps from $s_0=0$ are
\begin{equation}
 \Phi_1^A(x,y)=(x+1,y),\qquad
 \Phi_1^B(x,y)=(x+1,y+3(1-x^2)).
 \label{eq:appendix-shear-maps}
\end{equation}
Both systems report only $Z=\operatorname{pr}_x\Phi_1(x,y)+\xi=x+1+\xi$ with the same Gaussian noise. For the common target line $B=\{y=0\}$, $D_2\cap(\Phi_1^A)^{-1}(B)$ is one segment. In system B the preimage is $y=-3(1-x^2)$ and membership in $D_2$ requires
\begin{equation}
 9u^2-17u+5\le0,\qquad u=x^2.
 \label{eq:appendix-shear-roots}
\end{equation}
Its two positive roots split the admissible $x$ values into one negative and one positive interval, hence two disjoint graph arcs. The reported posterior and future observable agree, but channel multiplicity does not; the missing target label is therefore task-essential.

\section{Full proof of the mass-rank stability theorem}
\label{app:mass-rank-proof}

Fix $x\in X$ and write $A_x=\{y:\widehat f(y)\le\widehat f(x)\}$. The uniform index bound gives
\begin{equation}
 \{f\le f(x)-2\varepsilon_f\}
 \subseteq A_x
 \subseteq
 \{f\le f(x)+2\varepsilon_f\}.
 \label{eq:rank-proof-inclusions}
\end{equation}
Indeed, on the left $\widehat f(y)\le f(y)+\varepsilon_f\le f(x)-\varepsilon_f\le\widehat f(x)$; the right implication follows by reversing these inequalities.

For the upper comparison, total variation and Eq.~\eqref{eq:rank-proof-inclusions} give
\begin{align}
 \widehat\rho_{\widehat f}(x)
 &=\widehat\mu(A_x)\\
 &\le \mu(A_x)+\varepsilon_\mu\\
 &\le \mu\{f\le f(x)+2\varepsilon_f\}+\varepsilon_\mu\\
 &\le \rho_f(x)+\omega_f(\varepsilon_f)+\varepsilon_\mu.
 \label{eq:rank-proof-upper}
\end{align}
The added index interval has length $2\varepsilon_f$ and lies in the window centered at $f(x)+\varepsilon_f$. Similarly,
\begin{align}
 \widehat\rho_{\widehat f}(x)
 &\ge \mu(A_x)-\varepsilon_\mu\\
 &\ge \mu\{f\le f(x)-2\varepsilon_f\}-\varepsilon_\mu\\
 &\ge \rho_f(x)-\omega_f(\varepsilon_f)-\varepsilon_\mu,
 \label{eq:rank-proof-lower}
\end{align}
where the removed interval lies in the window centered at $f(x)-\varepsilon_f$. Thus
\begin{equation}
 |\rho_f(x)-\widehat\rho_{\widehat f}(x)|
 \le\varepsilon_\mu+\omega_f(\varepsilon_f).
 \label{eq:rank-proof-first-modulus}
\end{equation}
Interchanging $(f,\mu)$ and $(\widehat f,\widehat\mu)$ repeats both one-sided comparisons and independently gives
\begin{equation}
 |\rho_f(x)-\widehat\rho_{\widehat f}(x)|
 \le\varepsilon_\mu+\omega_{\widehat f}(\varepsilon_f).
 \label{eq:rank-proof-second-modulus}
\end{equation}
The total-variation distance and the uniform index norm are symmetric. Therefore either complete bound is valid, which justifies taking their minimum. Supremizing over $x$ proves Eq.~\eqref{eq:mass-rank-bound}. The argument yields $\omega(\varepsilon_f)$, improving the more conservative $\omega(2\varepsilon_f)$ obtained by centering a radius-$2\varepsilon_f$ window at $f(x)$.

If $x\in C_f^{m-b}$, then $\widehat\rho_{\widehat f}(x)\le\rho_f(x)+b\le m$, so $C_f^{m-b}\subseteq C_{\widehat f}^m$. The other inclusions in Eq.~\eqref{eq:mass-rank-interleaving} follow by exchanging the functions and shifting $m$. The extensions by $\varnothing$ and $X$ make these statements valid when $m\pm b$ leaves $[0,1]$, with no endpoint exception.

Applying $H_k(-;\F)$ produces natural transformations between the functors $\R\to\mathrm{Vect}_{\F}$ shifted by $b$. Because the set maps are literal inclusions, the two composites are the $2b$ structure maps; hence the modules are $b$-interleaved. For q-tame modules, algebraic stability gives $d_B\le d_I\le b$ \cite{Chazal2016}. The isometry theorem is not required. Plateaus cause no problem for the rank filtrations; they affect only their identification with exact-mass credible sets.

\section{Forecast, update and the regular-value criterion}
\label{app:isotopy}

Forecast and update play fundamentally different geometric roles. A deterministic forecast is
\begin{equation}
 \bar\mu_{t+\Delta}=(\Phi_{t,t+\Delta})_{\#}\mu_t,
 \label{eq:forecast}
\end{equation}
whereas assimilation of a new observation \(y\) with likelihood \(L_y\) gives
\begin{equation}
 \dd\mu_{t+\Delta}(x)
 =\frac{L_y(x)\,\dd\bar\mu_{t+\Delta}(x)}
 {\int_{\X}L_y(z)\,\dd\bar\mu_{t+\Delta}(z)}.
 \label{eq:bayes-update}
\end{equation}
The measure is updated first. The credible filtration, metric, persistence module and skeleton are then transported where justified and otherwise recomputed from the updated model. There is no meaningful operation that simply ``pushes forward the entire tuple'' through an observational update.

To identify when topology can change, let \(\A\) be a compact smooth manifold with $C^2$ boundary and consider a one-parameter family of reference-relative posterior densities \(r_\theta\), with corresponding regular mass threshold \(c_\theta(m)\). The following is a supporting continuation result in the spirit of standard isotopy and isotopy-extension lemmas for regular level families, including their boundary and stratified variants \cite{Hirsch1976,GoreskyMacPherson1988}. Its role here is to record the condition for a \emph{moving enclosed-mass threshold}; no independent novelty is claimed for the underlying isotopy argument.

\begin{theorem}[Regular-value criterion for topological stability]
\label{thm:isotopy}
Fix \(m\in(0,1)\). Suppose \((\theta,x)\mapsto r_\theta(x)\) and \(\theta\mapsto c_\theta(m)\) are \(C^2\), and for every \(\theta\in[0,1]\):
\begin{enumerate}[label=(\roman*)]
\item \(\nabla_x r_\theta(x)\neq0\) whenever \(r_\theta(x)=c_\theta(m)\) in the interior of \(\A\);
\item the level set intersects \(\partial\A\) transversally.
\end{enumerate}
Then there is a boundary-preserving ambient isotopy $\Psi_\theta:\A\to\A$ with $\Psi_0=\Id$ and
\[
 \Psi_\theta(C_0^m)=C_\theta^m,
 \qquad C_\theta^m=\{x\in\A:r_\theta(x)\ge c_\theta(m)\}.
\]
Thus the sets and their level boundaries are mutually ambiently isotopic; their homology groups and Betti numbers are constant in \(\theta\).
\end{theorem}
\begin{proof}

The threshold regularity need not be assumed independently in common cases. By the coarea formula, if the selected level is regular and its level surface has finite nonzero weighted area, the derivative of $G_\theta(c)=\mu_\theta\{r_\theta\ge c\}$ with respect to $c$ is finite and nonzero. Joint $C^2$ regularity and the implicit-function theorem then give the required differentiability of $c_\theta(m)$ locally; the theorem assumes a global branch only to avoid repeated continuation notation.

\begin{remark}[Role of the global hypotheses]
Compactness and $C^2$ boundary regularity are sufficient conditions for the global ambient-isotopy proof, not necessary conditions for every local continuation result. Properness on a neighborhood of the moving level set can replace global compactness, and stratified variants are possible under additional hypotheses. The nonsmooth, cut-defined admissible regions used in orbit determination do not satisfy this theorem automatically: it explains which critical or tangency mechanisms to test, but it does not certify their grid intersections.
\end{remark}

Set \(F(\theta,x)=r_\theta(x)-c_\theta(m)\). In the interior of $\A$, \(\nabla_xF\neq0\) on \(F^{-1}(0)\), and the local choice
\begin{equation}
 V_\theta(x)=-\frac{\partial_\theta F(\theta,x)}{\|\nabla_xF(\theta,x)\|^2}\nabla_xF(\theta,x).
 \label{eq:isotopy-vector}
\end{equation}
satisfies
\begin{equation}
 \partial_\theta F+\nabla_xF\cdot V_\theta=0.
 \label{eq:isotopy-transport}
\end{equation}
At $F^{-1}(0)\cap\partial\A$, transversality is equivalent to the nonvanishing of the tangential gradient $\nabla_{\partial\A}F$. Hence the boundary field
\[
 V_\theta^{\partial}(x)=
 -\frac{\partial_\theta F(\theta,x)}
 {\|\nabla_{\partial\A}F(\theta,x)\|^2}
 \nabla_{\partial\A}F(\theta,x)
\]
is tangent to $\partial\A$ and satisfies Eq.~\eqref{eq:isotopy-transport}. Extend these local interior and boundary solutions to a collar of the level set. Because Eq.~\eqref{eq:isotopy-transport} is affine in $V_\theta$, a partition-of-unity combination of local solutions still satisfies it. Multiplying by a cut-off that equals one near $F^{-1}(0)$ extends the field to all of $\A$ without changing the transport identity on the level set and makes the flow the identity away from a tubular neighborhood. The $C^2$ regularity assumptions make the resulting field $C^1$; compactness gives its parameter-dependent flow for all $\theta\in[0,1]$. This flow preserves $\partial\A$ and transports the zero level. A point cannot change from the superlevel side to the sublevel side without crossing that transported boundary, while the cut-off fixes the remote sign regions. Hence the ambient flow carries $C_0^m$ to $C_\theta^m$. Critical values or boundary tangencies are precisely the points where this construction can fail, consistent with Morse-theoretic intuition \cite{Milnor1963}.
\end{proof}

\section{Proof details for the local Gaussian limit}
\label{app:gaussian}

Use normal coordinates centered at \(x_\star\) and set \(x=x_\star+\sqrt\varepsilon z\). Uniformly on bounded \(z\)-sets,
\begin{align}
 V(x)&=V(x_\star)+\frac{\varepsilon}{2}z^{\T}Hz
 +O(\varepsilon^{3/2}\|z\|^3),\\
 a(x)&=a(x_\star)+O(\sqrt\varepsilon\|z\|).
\end{align}
After normalization, the rescaled density converges locally uniformly to the density of \(\mathcal N(0,H^{-1})\). For every chart neighborhood $U$ of $x_\star$, compactness and uniqueness of the global minimum give $\inf_{\A\setminus U}V>V(x_\star)$. The resulting positive gap makes the posterior mass outside $U$ exponentially small. Inside $U$, positive definiteness of $H$ gives a quadratic lower bound, so the rescaled tails and second moments are uniformly integrable. At a regular Gaussian density quantile, local uniform convergence and the tail bound imply convergence of the threshold and of the superlevel-set boundaries. The Gaussian superlevel set carrying mass \(m\) is \(\{z:z^{\T}Hz\le\chi^2_{n,m}\}\), proving \eqref{eq:appendix-gaussian-set}; uniform integrability gives \eqref{eq:appendix-gaussian-covariance}.

\section{Coordinate form of the Gaussian metric}
\label{app:metric-coordinate}

Let \(J_{t,s}(x)=D\Phi_{t,s}(x)\) and \(H_s(x_s)=Dh_s(x_s)\). Equation~\eqref{eq:gramian} becomes
\begin{equation}
 g_{t,T,x}=\int_t^{t+T}
 J_{t,s}(x)^{\T}H_s(\Phi_{t,s}(x))^{\T}R_s^{-1}
 H_s(\Phi_{t,s}(x))J_{t,s}(x)\,\dd s.
 \label{eq:expanded-metric}
\end{equation}
Under \(x'=\psi(x)\), the matrix representative satisfies
\begin{equation}
 G'(x')=D\psi_x^{-\T}G(x)D\psi_x^{-1},
 \label{eq:matrix-transform}
\end{equation}
which is the covariant transformation law of a \((0,2)\)-tensor.

\section{The equal-moment mixture witness}
\label{app:covariance-witness}

Let $\mu_1=\mathcal N(0,I_n)$. Every density superlevel set of $\mu_1$ is a ball, so $\beta_0(C_1^m)=1$ at every regular mass. For $0<a<1$, let
\[
 \mu_2=\tfrac12\mathcal N(ae_1,\Sigma_a)
       +\tfrac12\mathcal N(-ae_1,\Sigma_a),
 \]
 \begin{equation}
 \Sigma_a=\operatorname{diag}(1-a^2,1,\ldots,1)
 \label{eq:equal-cov-mixture}
\end{equation}
The symmetry gives zero mean and
\(
 \operatorname{Cov}(\mu_2)=\Sigma_a+a^2e_1e_1^{\T}=I_n,
\)
so $\mu_1$ and $\mu_2$ have identical first two moments. Both densities are smooth and strictly positive on all of $\R^n$; the disconnected credible region is therefore not caused by disconnected support. For the one-dimensional factor along $e_1$, stationary points solve
\(
 x=a\tanh\!\left(ax/(1-a^2)\right).
\)
Indeed, up to a constant this factor is
\[
 p_a(x)=
 \exp\!\left[-\frac{x^2+a^2}{2(1-a^2)}\right]
 \cosh\!\left(\frac{ax}{1-a^2}\right)
 \]
 \[
 (\log p_a)''(0)=\frac{2a^2-1}{(1-a^2)^2}.
\]
Thus the classical equal-weight normal-mixture criterion $a^2>1/2$, equivalently $a>2^{-1/2}\simeq0.70710678$ \cite{RobertsonFryer1969,Behboodian1970}, makes the origin a strict local minimum; there is exactly one positive stationary solution and its negative, both nondegenerate maxima. For $n>1$ the density is the product $p_a(x_1)\phi(z)$. Radial contraction in $z$ deformation-retracts every superlevel set onto $\{p_a\ge c/\phi(0)\}\times\{0\}$, so its zeroth Betti number is exactly the one-dimensional count. Hence any regular threshold strictly between the density at the origin and the modal density has two superlevel components. Choosing $m$ as the enclosed probability gives $\beta_0(C_2^m)=2$, while $\beta_0(C_1^m)=1$.

\section{Complexity and high-dimensional representations}
\label{app:complexity}

Let $N$ be the number of weighted posterior samples, $n$ the state dimension, $q$ the output dimension, $k$ a sparse-graph degree and $C_{\Phi}$ the cost of one propagation. Sorting by the log-density index costs $O(N\log N)$. For the $H_0$ questions used here, a sparse adjacency graph can be processed incrementally by union--find at near-linear cost in its edges; higher-dimensional Vietoris--Rips complexes are not the default and may grow combinatorially.

Terminal or hitting labels cost $O(NC_{\Phi})$ before boundary refinement. Direct Fisher evaluation requires tangent or score calculations, while adjoints are preferable when $q\ll n$ and matrix-free low-rank methods when the information operator has small effective rank. Coherent-set, committor and invariant-manifold errors are system-specific and must remain separate from posterior approximation error.

High-dimensional reports should use posterior-informed subspaces, sparse graphs and output-adapted metrics rather than reconstructing every ambient credible surface. The required validation remains componentwise: posterior mass and score, graph or mesh scale, future-law approximation, event labels and reduced-output error.

\printcredits

\bibliographystyle{cas-model2-names}
\bibliography{dug_physica_d}

\end{document}